\documentclass[reprint,pra,
superscriptaddress,
 amsmath,amssymb,
 aps,
 onecolumn,
]{revtex4-2}
\usepackage{graphicx}% Include figure files
\usepackage{dcolumn}% Align table columns on decimal point
\usepackage{bm}% bold math
\usepackage[T1]{fontenc} %256 bit font encoding
\usepackage[english]{babel} %english language
\usepackage[utf8]{inputenc}
\usepackage[table,usenames,dvipsnames]{xcolor}
\usepackage[psdextra]{hyperref}
\usepackage{bookmark}
\usepackage{amsfonts}
\usepackage{amsmath,amssymb,amsfonts,mathtools,amsthm}
\usepackage{graphicx}
\usepackage{bbm}
\usepackage{tikz}
\newcommand{\1}{\mathbb{1}}

\DeclareMathOperator{\U}{U}
\DeclareMathOperator{\tr}{tr}
\DeclareMathOperator{\C}{\mathcal C}
\DeclareMathOperator*{\E}{\mathop{\mathbb{E}}}
\DeclareMathOperator{\andd}{\, \& \,}

\DeclareMathOperator{\brick}{brick}
\DeclareMathOperator{\choi}{Choi}
\DeclareMathOperator{\DW}{DW}
\DeclareMathOperator{\Int}{Int} 
\DeclareMathOperator{\id}{Id}

\DeclareMathOperator{\QEC}{QEC} 
\DeclareMathOperator{\AQEC}{AQEC} 

\usetikzlibrary{quantikz2, calc, positioning}
\usepackage{enumitem}
\usepackage[normalem]{ulem}
\usepackage{hyperref}
\newtheorem{lemma}{Lemma}
\newtheorem{corollary}{Corollary}
\newtheorem{definition}{Definition}
\newtheorem{proposition}{Proposition}
\newtheorem{theorem}{Theorem}

\newcommand{\Haar}{\mathrm{Haar}}

\usepackage{bbold}
\renewcommand{\1}{\mathbb{1}}

\begin{document}
\title{Error correction with brickwork Clifford circuits}
\author{Twan Kroll}
\affiliation{Institute for Logic, Language and Computation, University of Amsterdam, Amsterdam, Netherlands}
\author{Jonas Helsen}%
\affiliation{Qusoft and Centrum Wiskunde \& Informatica, Amsterdam, Netherlands}%
\date{\today}
\begin{abstract}
We prove that random 1D Clifford brickwork circuits form (in expectation) good approximate quantum error correction codes in logarithmic depth. Our proof makes use of the statistical mechanics techniques for random circuits developed by Dalzell et al. [PRX Quantum 3, 010333], adapted extensively to our own purpose. We also consider exact error correction, where we give matching upper and lower bounds for the required depth in which random 1D Clifford brickwork circuits become error correcting. 
\end{abstract}

\maketitle
\section{Introduction}
Quantum error correction is a key building block for large scale quantum computers. By encoding a number of logical qubits into a larger number of physical qubits, one obtains protection against the error seemingly inherent in the operation of physical qubits. Since Shor's development of the first error reducing quantum codes \cite{shor1995scheme}, a vast literature has been developed over the past three decades, constructing error correction codes with various beneficial properties such as low-weight checks, high encoding rate, high distance, and compatibility with spatial locality (with trade-offs existing between all these parameters). \\

In its exact form, QEC requires perfect recovery of an encoded state after suffering some unknown error and is characterized by the Knill-Laflamme conditions \cite{knill2000theory}.
Relaxing the perfect recovery requirement by an approximate recovery criterion might be more suitable for practical purposes as codes that achieve approximate QEC can outperform exact error correcting codes in different ways \cite{leung1997approximate}. In fact, limits on perfect recovery do not necessarily preclude approximate recovery guarantees \cite{crepeau2005approximate}. A popular approximate QEC notion is to require high entanglement fidelity; for which the Knill-Laflamme conditions can be generalized \cite{Beny2010} and is closely related to the average fidelity between the recovered state and logical state \cite{horodecki1999general, gilchrist2005distance,Kong2022}. \\

In analogy to the fruitful study of classical (linear) random codes, there has been substantial interest recently in \emph{random} (stabilizer) quantum error correction codes. Already in 2013 Brown \& Fawzi proved that Clifford circuits of depth $O(\log^3(n))$ in an all-to-all connected architecture form good (i.e. linear rate and distance) exact quantum error correction codes \cite{Brown}. Such codes however are not geometrically restricted and far from LDPC. More recently Gullans et al. \cite{Gullans_2021} provided numerical and analytical evidence that \emph{geometrically local} random Clifford circuits are good approximate quantum error correction codes already at logarithmic depth. This work was very recently supplemented with an analytical proof \cite{liu2025approximatequantumerrorcorrection} that a certain class of $1$D geometrically local circuits, often called the superbrick architecture \cite{schuster}, are good approximate quantum error correction codes. The analysis of these superbrick circuits was inspired by the surprising finding in Schuster et al. \cite{schuster}, that such unitaries form so-called unitary $k$-designs (thus mimicking statistical aspects of fully random unitaries). However in the work of Liu et al. \cite{liu2025approximatequantumerrorcorrection}, the error correcting nature of the more natural architecture of brickwork random circuits (which is also the architecture discussed in Gullans et al., see Fig \ref{fig:brickwork} for an illustration) was left open~\footnote{It must be noted here that the authors of \cite{liu2025approximatequantumerrorcorrection} have very recently provided an updated version of their work that also includes bounds on the approximate error correction abilities of log-depth \emph{brickwork} circuits \cite{liu2026approximatejournal}. Our results were arrived at independently.}. The main result of this work is to close this gap. To do this we lean on the statistical mechanics techniques developed in \cite{Dalzell_2022}, which were also recently used to prove that these brickwork random circuits constitute state $2$-designs \cite{heinrich2025anti}.  Random circuits have a wide range of applications. These include randomized benchmarking \cite{helsen2022general, knill2008randomized, heinrich2022randomized}, constructing unitary designs \cite{schuster, brandao2016local, haferkamp2022random}, quantum algorithms \cite{bertoni2024shallow}, and many-body and black-hole physics \cite{nahum2017quantum, susskind2020second, akers2024holographic}. This work contributes to the broader endeavor of understanding random quantum circuits and their various applications.

\section{Results}
Our main result is a proof that 1D brickwork Clifford circuits become approximate quantum error correction codes in depth $O(\log(n))$.  Specifically we prove the following theorem for the ensemble  $\chi_{\brick}^{D(n)}$ of $n$-qubit random brickwork circuits of depth $D(n)$ consisting of $2$-qubit Clifford gates.
\begin{theorem}[Informal]\label{thm:informal} For appropriate constants $0<r<1$ and $c>\frac{1}{|\log r|}$ with $D(n)=c\log(n)$, the expected Choi error can be bounded for appropriate error thresholds as follows:
\begin{equation}
    \E_{U\sim \chi_{\brick}^{(D)}}\left[\epsilon_{\choi}(U)\right]\leq \mathcal O\left(n^{\frac{1-c|\log(r)|}{4}}\right).
\end{equation}
\end{theorem}
From Theorem~\ref{thm:informal} it is clear that the expected Choi error tends to zero in the limit $n\to \infty$ with depth $\mathcal O(\log(n))$.
To prove this statement we adapt the statistical mechanics methods used by Dalzell et al. \cite{Dalzell_2022} to prove that log-depth circuits anti-concentrate. In this statistical physics model, we have \emph{domain wall trajectories} as depicted in Figure~\ref{fig:domainwalltrajectory}. In these trajectories a \emph{domain wall} follows a certain path subject to some constraints. When two paths meet, the domain walls annihilate. Furthermore, a trajectory has a weight associated with it (decreasing exponentially in the length of the path). To prove Theorem~\ref{thm:informal} we first prove that the expected Choi error can be upper bounded by summing over the weight of all possible allowed trajectories. 

Next, we upper bound this sum of weights over all allowed trajectories using combinatorial arguments and properties of said trajectories. One such important property is that the weight of a trajectory, as in Figure~\ref{fig:domainwalltrajectory}, can be decomposed as a product over the weights of the individual paths. We group the paths into two categories, surviving paths and annihilating paths. The surviving domain walls partition the space into disjoint intervals. The intervals vary in size over time as can be seen in Figure~\ref{fig:domainwalltrajectory}. Then, by upper bounding the weight of all possible trajectories that can occur within an interval and using the disjointness of these intervals we arrive at an upper bound over all possible trajectories. For the details of this argument we refer to Section~\ref{sec:logarithmicdepth}.\\

\begin{figure}
    \centering
\includegraphics[scale = 0.8]{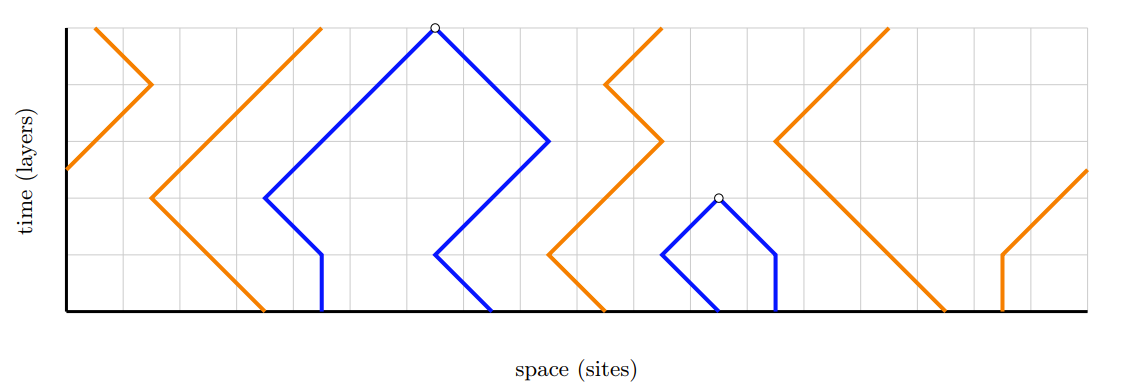}
\caption{Trajectories in the stat-mech model. Paths labeled blue annihilate while the orange paths survive. Each path carries with it a weight that must be controlled in order to prove a bound on the expected Choi error (i.e. entanglement infidelity). Note that the surviving domain walls partition the qubits into "independent" intervals. This insight is key to achieving theorem \ref{thm:informal}. For the sake of readability of this figure, we apply all gates in a single layer in the same time step, while in the proof of theorem \ref{thm:informal}, to avoid ambiguities, we will apply them one-by-one in a specified order (in effect moving only one domain wall at a time).  }
\label{fig:domainwalltrajectory}
\end{figure}

Our second result provides a tight characterization of the \emph{exact} error correcting properties of random brickwork circuits. Due to light-cone constraints, any $U\sim \chi_{\brick}^{D(n)}$ can only achieve a code distance on the order of $\mathcal O(D(n))$. Conversely, we prove that any code distance $d(n)$ is achieved in depth $\mathcal O(d(n))$ for any $U\sim \chi_{\brick}^{D(n)}$ with high probability.

\begin{theorem}[Informal]\label{thm:informal2}
    Let $d(n)\in o(n)\cap \Omega(\log(n))$. Then there exists a $D(n)\in \mathcal O(d(n))$ such that the probability that $U\sim \chi_{\brick}^{D(n)}$ is a $[[n,k,d(n)+1]]$ error correcting code tends to 1 in the limit $n\to \infty$ for any constant rate $\frac{k}{n}$.
\end{theorem}
The proof of Theorem~\ref{thm:informal2} rests largely on the analysis used for approximate error correction combined with the work of Brown \& Fawzi \cite{Brown} (in particular their conditions for error correction for Clifford encoding circuits). Surprisingly we achieve a somewhat stronger result than the approximate version, as we show that error correction is possible almost surely (though at much higher depth), as opposed to merely in expectation in the approximate case. We believe this to be a limitation of the proof technique (in particular the use of only second moments), which can be overcome with a more sophisticated analysis.

\section{Preliminaries}
In this section we recap some known notions and set notation that will be required later.
\subsection{(Approximate) error correction}
We give a brief recap of approximate quantum error correction.  More specifically,  we introduce the \emph{Choi error}, the metric used to quantify the error correcting performance of a code and consider two noise models of interest. The exposition mirrors Sections B and C of Chapter 3 of \cite{liu2025approximatequantumerrorcorrection}.\\

Let $L$ denote the register of $k$ logical qubits that we aim to protect. To protect this information, the logical qubits are encoded in a $n$-qubit physical system, $S$, through an encoding channel $\mathcal E_{L\to S}$. The logical qubits may also be entangled with a $k$-qubit reference system $R$. To model the error, the physical system will undergo a noise channel $\mathcal N$ after encoding. Lastly, to recover the information, a decoding channel, $\mathcal D_{S\to L}$ is applied. If the initial state is $\rho_{LR}$, then the recovered state is:
\begin{equation}
    \left[(\mathcal D\circ \mathcal N \circ \mathcal E)\otimes \id_R\right](\rho_{LR}).
\end{equation}
Note that a unitary $U\in \U(2^n)$ defines a \emph{unitary encoder}: $\mathcal E_{L\to S}(\rho_L)=U\left(\rho_L\otimes |\boldsymbol 0\rangle\langle \boldsymbol 0|_{S\setminus L}\right)U^\dagger$. We now define the quantity of interest.

\begin{definition}
    Let $|\hat \phi\rangle_{LR}$ denote the maximally entangled state between $L$ and $R$. Given an encoder $\mathcal E$ and a noise channel $\mathcal N$, the \emph{Choi fidelity} is given by:
    \begin{equation}\label{eq:choifidelity}F_{\choi}=\max_{\mathcal D}F\left(|\hat \phi\rangle \langle \hat \phi|_{LR},\left[(\mathcal D\circ \mathcal N \circ \mathcal E)\otimes \id_R\right](|\hat \phi\rangle \langle \hat \phi|_{LR})\right).\end{equation}
    The \emph{Choi error} is then defined as:
    \begin{equation}\epsilon_{\choi}=\sqrt{1-F_{\choi}^2}=\min_{\mathcal D}P\left(|\hat \phi\rangle \langle \hat \phi|_{LR},\left[(\mathcal D\circ \mathcal N \circ \mathcal E)\otimes \id_R\right](|\hat \phi\rangle \langle \hat \phi|_{LR})\right),
    \label{eq:choi}\end{equation}
    where $P(\cdot, \cdot)$ is the purified distance.
\end{definition}

The Choi fidelity quantifies in the optimal case how well the code preserves the entanglement. An alternative metric for performance is the average fidelity.

A useful tool to determine the Choi error is the formalism of \emph{complementary channels} \cite{Beny2010,Kong2022}. Using the Stinespring representation of a channel, there exists an environment system $E$ such that for any $\rho_S\in D(\mathcal H_S)$ the following holds:
\begin{equation}
    \mathcal N(\rho_S) = \tr_E\left(V\rho_S V^\dagger\right),
\end{equation}
for some isometry $V\in \mathcal L(\mathcal H_S,\mathcal H_S\otimes \mathcal H_E)$.

For random circuits, the unitary encoding is sampled from some distribution $\chi$. In this case, the expected Choi error is of interest. The following proposition states a useful upper bound for the expected Choi error which is derived in \cite{liu2025approximatequantumerrorcorrection}.

\begin{proposition}\label{prop:expectedchoierror}Define $\rho_{SR}=|\hat \phi\rangle\langle \hat \phi|_{LR}\otimes |0\rangle\langle 0|^{\otimes (n-k)}_{S\setminus L}$, $\tilde{\mathcal N}=\tau_E^{-\frac14}\hat{\mathcal N}\tau_E^{-\frac14}$ with $\tau_E=\hat{\mathcal N}_{S\to E}(\frac{1}{2^n}I_S)$ being the reduced state on the subsystem $E$ of $\tau_{SE}=(\id_S\otimes\hat{\mathcal N}_{S'\to E})(|\hat \phi\rangle \langle \hat \phi|_{SS'})$, where $S'$ is a copy of $S$. Let $\chi$ denote a distribution of unitaries on $n$ qubits that is a unitary $1$-design, then the Choi error satisfies:
    \begin{equation}\E_{U\sim \chi}\left[\epsilon_{\choi}\right]\leq\sqrt{\sqrt{2^k\E_{U\sim \chi}\left[\tr\left(\left(\tilde{\mathcal N}(U_S\rho_{SR}U_S^\dagger)\right)^2\right)\right]-1}}.\end{equation}
\end{proposition}

\subsection{Noise models}
We will consider two noise models, \emph{Pauli noise} and the \emph{erasure noise}. The Pauli noise channel on a single-qubit is defined as
\begin{equation}
    \mathcal N_{P}(\rho)=p_I\rho +p_XX\rho X+p_Y Y\rho Y + p_Z Z\rho Z,
\end{equation}
and shall be referred to as strength-$\vec p$ Pauli noise. On all qubits, the Pauli noise channel is applied independently, giving the overall noise channel $\mathcal N=\mathcal N_{P}^{\otimes n}$.

Meanwhile, erasure noise for a single-qubit is defined as:
\begin{equation}
    \mathcal N_e(\rho)=(1-p)\rho + p|2\rangle \langle 2|,
\end{equation}
where the $|2\rangle$ represents the erasure and $p\in [0,1]$ is the probability at which the erasure occurs. Over $n$ qubits, we shall assume that each qubit faces an i.i.d. erasure error with probability $p$ giving an overall noise channel: $\mathcal N=\mathcal N_e^{\otimes n}$.

\subsection{Designs}
In this section we define the concept of unitary designs and explicitly determine the second moment with respect to the Haar measure. For a proper introduction to these notions we refer to \cite{mele2024introduction}. The Haar measure on $\U(d)$ shall be denoted by $\Haar(d)$.

\begin{definition}\label{def:unitarydesign}
    Consider a probability distribution $\nu$ over a set of unitaries $\U(d)$. We call $\nu$ a $t$-design if and only if:
    \begin{equation}\label{kdesign}
        \E_{V\sim \nu}\left[V^{\otimes t}O(V^\dagger)^{\otimes t}\right]=\E_{U\sim \Haar(d)}\left[U^{\otimes t}O(U^\dagger)^{\otimes t}\right],
    \end{equation}
    for any $O\in \mathcal L((\mathbb C^d)^{\otimes t})$.
\end{definition}

The uniform distribution over the $n$-qubit Clifford gates is a well-known $2$-design \cite{Dankert2009} and shall be denoted by $\chi^{(n)}_{\mathcal C}$.
The expression $\E_{U\sim \Haar(d)}\left[U^{\otimes t}O(U^\dagger)^{\otimes t}\right]$ is known as the $t$-th order moment of the Haar distribution. When $t=2$, this expression turns out to be a linear combination over $\{I,S\}$, where $S$ is the swap operator. The exact expression of these coefficients is a standard result and will be stated in the next lemma without proof.

\begin{lemma}[Second Order Moment]
    \label{lemma:haartwirl}
    Let $O\in \mathcal L(\mathbb C^d\otimes \mathbb C^d)$ be arbitrary and $X:=\E_{U\sim \Haar(d)}\left[U^{\otimes 2}O(U^\dagger)^{\otimes 2}\right]$. Then, $X=\alpha I +\beta S$ where:
    \begin{equation}\alpha = \frac{1}{(d^2-1)}\tr(O)-\frac{1}{d(d^2-1)}\tr(SO),\qquad \beta = \frac{1}{(d^2-1)}\tr(SO)-\frac{1}{d(d^2-1)}\tr(O).\end{equation}
\end{lemma}
\subsection{Brickwork circuits}\label{sec:brickworkcircuits}
We consider a standard 1D brickwork random Clifford circuit architecture of $D$ layers with periodic boundary conditions. In these $D$ layers uniformly sampled $2$-qubit Clifford gates are arranged in a brickwork like in Figure~\ref{fig:brickwork}. Furthermore, given a rate of $\frac{a}{b}$, we consider $m$ blocks of $a$ logical qubits and $b-a$ ancilla qubits. Let $\chi_{\brick}$ denote the distribution of such random circuits. Schematically, we refer to Figure~\ref{fig:brickwork} for $m=2$ and $\frac{a}{b}=\frac{2}{5}$. The total number of logical qubits is then $k=am$ whereas the number of physical qubits is $n=bm$.

\begin{figure}
    \centering
    \includegraphics[width=0.4\linewidth]{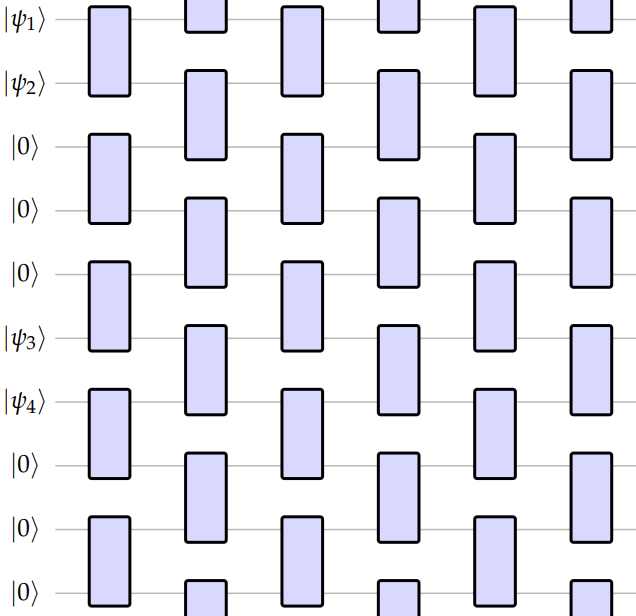}
    \caption{A diagram showing the architecture when the rate is $\frac{2}{5}$ and there are 2 blocks in depth $6$, i.e., $n=10$. Each gate is a uniformly sampled $2$-qubit Clifford gate.}
    \label{fig:brickwork}
\end{figure}
More formally, we denote a random circuit in this architecture as:
\begin{equation}\label{eq:RQCexplicitform}
    U_S=U_{A_1}U_{A_2}\dots U_{A_s},
\end{equation}
Here, $W_{X}$ denotes that $W$ acts on the qubits contained in $X$, where $|A_i|=2$ and contains only adjacent qubits since we assume 1D connectivity with periodic boundary conditions (PBC). When it is needed to apply a layer of gates in order (as will be the case in the next section), we will apply them in increasing order, starting with the gate acting on qubits $1$ and $2$ in the first layer and qubits $2$ and $3$ in the second layer (repeating periodically in time).

\section{Logarithmic depth approximate error correction}\label{sec:logarithmicdepth}

This section is dedicated to proving the main result: random brickwork Clifford circuits form approximate error correcting codes in logarithmic depth w.r.t. the number of qubits for Pauli and erasure noise.
Our starting point is the upper bound for the expected Choi error by Proposition~\ref{prop:expectedchoierror}: 
\begin{equation}
    \E_{U_S\sim \chi_{\brick}}\epsilon_{\choi}\leq \sqrt{\sqrt{2^k\E_{U_S\sim \chi_{\brick}}\left[\tr\left(\left(\tilde{\mathcal N}(U_S\rho_{SR}U_S^\dagger)\right)^2\right)\right]-1}}:= \sqrt{\sqrt{Z-1}}.
    \label{eq:expectationterm}
\end{equation}
implicitly defining the \emph{partition function} $Z$, analogous to \cite{Dalzell_2022}. Using standard arguments from Weingarten calculus (and a fair bit of elbow grease), we can derive an explicit form for the partition function. We will postpone the derivation of this fact to the appendix.

\begin{lemma}[Partition function]\label{lem:partition}
    Let $\chi_{\brick}$ denote the distribution over random Clifford circuits with $D$ layers as defined above, with $s$ number of gates, i.e., $s=\frac{nD}{2}$. Consider the configuration space $\{I,S\}^n$ and let $\vec{\gamma}\in \{I,S\}^{n\times (s+1)}$ be a trajectory, where $\vec \gamma^{(t)}\in \{I,S\}$ is the configuration at time $t$. Lastly, if $U^{(t)}$ acts on qubits $\{x,x+1\}$, we define for $\gamma,\nu\in \{I,S\}^n$
    \begin{equation}\label{eq:momentmatrix}
        M^{(t)}_{\nu\gamma}=\begin{cases}
            1, & \text{if $\gamma_x=\gamma_{x+1}$ and $\gamma=\nu$}\\
            \frac{2}{5}, &\text{if $\gamma_x\neq \gamma_{x+1}$ and $\nu_x=\nu_{x+1}$ and $\gamma_y=\nu_y\; \forall y\in [n]\setminus\{x,x+1\}$}\\
            0,& \text{otherwise.}
        \end{cases}
    \end{equation}
    The partition function can now be expressed as follows:
    
    \begin{equation}
        Z= 3^k\cdot\frac{2^n}{3^n}\sum_{\substack{\vec \gamma\in \{I,S\}^{n\times (s+1)}\\ \vec \gamma^{(0)}\in (\{S\}^a\times\{I,S\}^{b-a})^m}}\lambda^{|\vec\gamma^{(s)}|}\prod_{t=1}^sM^{(t)}_{\vec \gamma^{(t)}\vec \gamma^{(t-1)}}, \label{eq:partitionfunction}
    \end{equation}
    where $\lambda=2^{f-1}$ and $|\vec \gamma^{(s)}|$ is the number of $S$ factors in $\vec \gamma^{(s)}$. 
\end{lemma}
In the above, $f$ is a parameter depending on the underlying noise channel. If we assume Pauli noise, then $\lambda=2^{f_P(\vec p)-1}$ with $f_P(\vec p)=2\log(\sqrt{p_I}+\sqrt{p_X}+\sqrt{p_Y}+\sqrt{p_Z})$, while for erasure noise $\lambda=2^{f_e(p)-1}$ with $f_e(p)=\log(1+3p)$. Note that both noise channels satisfy $0\leq f \leq 2$.

Note that from this lemma it follows that when $Z^{(s)}\to 1$ as $n\to \infty$ the expected Choi error tends to 0, thus achieving approximate error correction. The rest of this section will be dedicated to proving this latter convergence statement. As an intermediate step we show that the partition function decreases in the circuit depth $s$. 

\begin{lemma}\label{lem:decreasing}
    The partition function $Z^{(s)}$ decreases in the circuit size $s$.
\end{lemma}
\begin{proof}
    Consider a configuration trajectory $\vec \gamma$ of length $s$. Its current contribution is given by
    \begin{equation}
        3^k\cdot\frac{2^n}{3^n}\cdot\lambda^{|\vec\gamma^{(s)}|}\prod_{t=1}^sM^{(t)}_{\vec \gamma^{(t)}\vec \gamma^{(t-1)}}.
    \end{equation}
    When adding a uniformly sampled Clifford gate on a pair of qubits there are two possibilities. If $A_{s+1}=\{a,b\}$ and these bits agree in $\vec \gamma^{(s)}$, then $\vec \gamma^{(s+1)}=\vec \gamma^{(s)}$ and the contribution will not change. For each trajectory where these bits do not agree in $\vec \gamma^{(s)}$ we obtain two new trajectories. One in which both bits become $I$ and one in which they become $S$. The contribution of both these trajectories is
    \begin{align}
        \space &\frac{2}{5}\cdot 3^k\cdot\frac{2^n}{3^n}\cdot\lambda^{|\vec\gamma^{(s)}|+1}\prod_{t=1}^sM^{(t)}_{\vec \gamma^{(t)}\vec \gamma^{(t-1)}}+\frac{2}{5}\cdot 3^k\cdot\frac{2^n}{3^n}\cdot\lambda^{|\vec\gamma^{(s)}|-1}\prod_{t=1}^sM^{(t)}_{\vec \gamma^{(t)}\vec \gamma^{(t-1)}}\\
        &=\frac{2}{5}3^k\cdot\frac{2^n}{3^n}\cdot\lambda^{|\vec\gamma^{(s)}|}\prod_{t=1}^sM^{(t)}_{\vec \gamma^{(t)}\vec \gamma^{(t-1)}}\left(\frac25\lambda+\frac25\lambda^{-1}\right).\nonumber
    \end{align}
    Thus, the difference is a factor of $\frac25\lambda+\frac25\lambda^{-1}$, which for $0\leq f \leq 2$ satisfies $\frac25\lambda+\frac25\lambda^{-1}\leq 1$, showing that $Z^{(s)}$ is decreasing in $s$.
\end{proof}

Next we will argue that if we consider a brickwork circuit of infinite depth, the partition function has a clean functional form:

\begin{lemma}\label{lem:infinitedepth}
    Let $Z^{(\infty)}$ denote the partition function in infinite depth. Then, the following holds
    \begin{equation}
        Z^{(\infty)}=\frac{1}{1-2^{-2n}}\left(1-2^{-n+k}+2^{-n+k+nf}-2^{-2n+nf}\right).
    \end{equation}
\end{lemma}
\begin{proof} The proof follows by using the definition of the partition function backwards, and realising that a brickwork circuit of infinite depth is an exact Haar random unitary. 
First, we rewrite using the SWAP trick the term $\tr\left(\left(\tilde{\mathcal N}_{S\to E}(U_S\rho_{SR}U_S^\dagger)\right)^2\right)$. This gives us
\begin{align}
    \tr\left(\left(\tilde{\mathcal N}_{S\to E}(U_S\rho_{SR}U_S^\dagger)\right)^2\right)&=\tr\left(\left(\tilde{\mathcal N}_{S\to E}(U_S\rho_{SR}U_S^\dagger)\right)^{\otimes 2}(S_E\otimes S_R)\right)\\
    &=\tr\left(U_S^{\otimes 2}\rho_{SR}^{\otimes 2}(U_S^\dagger)^{\otimes 2}((\tilde{\mathcal N}_{S\to E}^{\dagger})^{\otimes 2}(S_E)\otimes S_R)\right)\\
    &=\tr_S\left(U_S^{\otimes 2} \tr_R\left(\rho_{SR}^{\otimes 2}S_R\right)(U_S^\dagger)^{\otimes 2}(\tilde{\mathcal N}_{S\to E}^{\dagger})^{\otimes 2}(S_E)\right).
\end{align}
Letting $\Haar(2^n)$ denote the Haar measure on the unitary group $\U(2^n)$, we see by invariance of the Haar measure that
\begin{align}
    &\E_{U_S\sim \Haar(2^n)}\left[\tr\left(\left(\tilde{\mathcal N}_{S\to E}(U_S\rho_{SR}U_S^\dagger)\right)^2\right)\right]=\E_{U_S\sim \Haar(2^n)}\left[\tr_S\left(U_S^{\otimes 2} \tr_R\left(\rho_{SR}^{\otimes 2}S_R\right)(U_S^\dagger)^{\otimes 2}(\tilde{\mathcal N}_{S\to E}^{\dagger})^{\otimes 2}(S_E)\right)\right]\\
    &=\E_{\substack{U_S\sim \Haar(2^n)\\ V_S\sim \Haar(2^n)}}\left[\tr_S\left((V_SU_S)^{\otimes 2} \tr_R\left(\rho_{SR}^{\otimes 2}S_R\right)((V_SU_S)^\dagger)^{\otimes 2}(\tilde{\mathcal N}_{S\to E}^{\dagger})^{\otimes 2}(S_E)\right)\right]\\
    &=\tr\left(\E_{U_S\sim \Haar(2^n)}\left[U_S^{\otimes 2}\tr_R\left(\rho_{SR}^{\otimes 2}S_R\right)(U_S^\dagger)^{\otimes 2}\right]\E_{V_S\sim \Haar(2^n)}\left[V_S^{\otimes 2}(\tilde{\mathcal N}_{S\to E}^{\dagger})^{\otimes 2}(S_E)(V_S^\dagger)^{\otimes 2}\right]\right),
\end{align}
where in the last equality we used the cyclicity of the trace. Using Lemma~\ref{lemma:haartwirl}, we find that 
\begin{equation}
    \E_{U_S\sim \Haar(2^n)}\left[U_S^{\otimes 2}\tr_R\left(\rho_{SR}^{\otimes 2}S_R\right)(U_S^\dagger)^{\otimes 2}\right]=\alpha I +\beta S, \; \text{and} \; \E_{V_S\sim \Haar(2^n)}\left[V_S^{\otimes 2}(\tilde{\mathcal N}_{S\to E}^{\dagger})^{\otimes 2}(S_E)(V_S^\dagger)^{\otimes 2}\right]=\alpha'I+\beta' S,
\end{equation} 
where $\alpha, \beta, \alpha'$ and $\beta'$ can be determined analogously as in the proof of Lemma~\ref{lem:partition}. This gives
\begin{align}
    \alpha =\frac{2^{-k}-2^{-n}}{2^{2n}-1}, \qquad  \beta 
    %=\frac{1}{2^{2n}-1}-\frac{2^{-k}}{2^n(2^{2n}-1)}
    =\frac{1-2^{-k-n}}{2^{2n}-1}, \qquad \alpha'=\frac{2^{2n}-2^{nf}}{2^{2n}-1}, \qquad \beta'=\frac{2^{nf+n}-2^n}{2^{2n}-1}.
\end{align}

Putting it all together and simplifying, one obtains
\begin{align}
    Z^{(\infty)}&=2^k\E_{U_S\sim \Haar(2^n)}\left[\tr\left(\left(\tilde{\mathcal N}_{S\to E}(U_S\rho_{SR}U_S^\dagger)\right)^2\right)\right]\\
    &=(\alpha\alpha'+\beta\beta')2^{2n+k}+(\alpha\beta'+\alpha'\beta)2^{n+k}\\
    &=\frac{1}{1-2^{-2n}}\left(1-2^{-n+k}+2^{-n+k+nf}-2^{-2n+nf}\right).
    %\\&=\frac{\left((2^{-k}-2^{-n})(2^{2n}-2^{nf})+(1-2^{-k-n})(2^{nf+n}-2^n)\right)2^{2n+k}+\left((1-2^{-k-n})(2^{2n}-2^{nf})+(2^{-k}-2^{-n})(2^{nf+n}-2^n)\right)2^{n+k}}{(2^{2n}-1)^2}
\end{align}
\end{proof}

It remains to prove our central result Theorem~\ref{thm:mainresultthesis}. To do this we turn to an alternative interpretation of the partition function in terms of \emph{domain walls}, following \cite{Dalzell_2022}. 
\begin{definition}
     Consider a configuration $\gamma\in \{I,S\}^n$ and define $\DW(\gamma) = g=\{e\in \{0,...,n-1\}\mid \gamma_{e}\neq \gamma_{e+1}\}$ to be the \emph{domain wall configuration} of $\gamma$. For a trajectory $\vec \gamma\in \{I,S\}^{n\times (s+1)}$ we define the \emph{domain wall trajectory} of $\vec \gamma$ as
     \begin{equation}
         \DW(\vec \gamma)=\left(\DW(\vec \gamma^{(t)})\right)_{t\in \{0,...,s\}}.
     \end{equation}
     We will sometimes denote it as $G=(g^{(t)})_{t\in \{0,...,s\}}$. We identify $\gamma_n$ and $\gamma_0$ for $\gamma\in \{I,S\}^n$ due to periodic boundary conditions.
\end{definition}
Let $\mathcal G^{(s)}=\DW(\{I,S\}^{n\times (s+1)})$ be the set of all domain wall trajectories. Note that the function $\DW$ is not injective. Each domain wall configuration $g$ corresponds to two possible configurations in $\{I,S\}^n$ and fixing $\gamma_0=I$ or $\gamma_0=S$ determines all other sites given $g$. This 2-to-1 correspondence also extends to trajectories. However, for us, any configuration trajectory $\vec \gamma$ starts in the set of initial configurations $\mathcal I := (\{S\}^a\times\{I,S\}^{b-a})^m$, with $\gamma^{(0)}_0=S$ (as $a>0$). This means we have that all the configuration trajectories with nonzero contribution to the partition function in Equation \eqref{eq:partitionfunction} are contained in
\begin{equation}
\mathcal T^{(s)}=\{\vec \gamma\in \{I,S\}^{n\times (s+1)} \mid \vec \gamma^{(0)}\in \mathcal I\}.
\end{equation}
We see that $\DW$ restricted to $\mathcal T^{(s)}$ is injective. The image of $\mathcal T^{(s)}$ under $\DW$ is given by
\begin{align}
    \mathcal G_{\mathcal B}=\{G\in \mathcal G^{(s)}\mid g^{(0)}\in \mathcal B\}, \; \text{where}\; \mathcal B=\{g\subseteq \{0,...,n-1\}\mid \exists\gamma \in \mathcal I : \DW(\gamma)=g\}. 
\end{align}
Thus $\DW:\mathcal T^{(s)}\to \mathcal G_{\mathcal B}$ is a bijective function. We will use $\Gamma:\mathcal G_B\to \mathcal T^{(s)}$ to denote the inverse of $\DW$. Using this we can express the partition function in terms of domain walls
\begin{equation}
    Z^{(s)}=3^k\cdot\frac{2^n}{3^n}\sum_{\substack{\vec \gamma\in \{I,S\}^{n\times (s+1)}\\ \vec \gamma^{(0)}\in (\{S\}^a\times\{I,S\}^{b-a})^m}}\lambda^{|\vec\gamma^{(s)}|}\prod_{t=1}^sM^{(t)}_{\vec \gamma^{(t)}\vec \gamma^{(t-1)}}=3^k\cdot\frac{2^n}{3^n}\sum_{G\in \mathcal G_{\mathcal B}}\lambda^{|\Gamma(G)^{(s)}|}W(G),
\end{equation}
where 
\begin{align}
    W(G)&:=\prod_{t=1}^sM^{(t)}_{g^{(t)}g^{(t-1)}}\\
    M^{(t)}_{g^{(t)}g^{(t-1)}}&:=\begin{cases}
        \frac25, & \text{if $\min A_t\in g^{(t-1)}$}\\
        1, & \text{otherwise}.
    \end{cases}
\end{align}

A domain wall trajectory can be partitioned into surviving domain walls and domain walls that annihilate, giving us:
\begin{equation}
    G=G_U\sqcup G_0=\left(g_U^{(0)} \sqcup g_0^{(0)},...,g_U^{(s)} \sqcup g_0^{(s)}\right),
\end{equation}
where $G_0$ is the domain wall configuration of the domain walls that annihilate, that is, $|g_0^{(s)}|=0$ and $G_U$ is the domain wall trajectory of the surviving domain walls. Note that $G_U$ has a conserved number of domain walls throughout its entire trajectory. Lastly, since $G_0$ and $G_U$ are disjoint, we find for the weights that $W(G)=W(G_0)W(G_U)$. Next, let $\mathcal G_0$ and $\mathcal G_U$ be the subsets of $\mathcal G^{(s)}$ of domain wall trajectories such that all domain walls annihilate and no annihilations occur, respectively. Similarly, we define $\mathcal G_{\mathcal B,0}$ as the subset of $\mathcal G_{\mathcal B}$ of annihilating domain wall trajectories starting from $\mathcal B$. More specifically, let $\mathcal G_{U,k}$ be the subset of $ \mathcal G_U$ with $k$ domain walls and remark that due to PBC, $\mathcal G_{U,k}$ is non-empty if and only if $k$ is even. \\

Now, we turn our attention to the dynamics of our partition function in terms of the domain wall picture. Consider $G\in \mathcal G_{\mathcal B}$, and decompose it in annihilating and surviving domain walls: $G=G_0\sqcup G_U$. Then we note that the $G_U$ partitions the $n$ sites into intervals whose sizes can change at each time step. 

Furthermore, all sites in an interval will have the same state, either $S$ or $I$. The goal is to express the partition function in terms of these partitions. With this in mind, we partition the $n$ sites in the following way. For $t\in \{0,...,s\}$ we let:
\begin{equation}
\mathcal P_t(G_U)=\{(x,x_{\mathrm{next}})\mid x\in g_U^{(t)} \},
\end{equation}
where $x_{\mathrm{next}}$ indicates the location of the domain wall immediately following the domain wall at location $x$ (wrapping around due to periodic boundary conditions). 
This specifies the intervals of a domain wall configuration $G_U$. 
Note that the size of an interval $p = (x,y)\in \mathcal P_t(G_U)$ is given by $b-a - 1\mod n$. 
Given a $p\in \mathcal P_0(G_U)$, we let $S_p^{(t)}$ denote the size of this interval at the time step $t$ and see that $S_p^{(t)}+1$ is the number of sites in the interval. Given $G=G_0\sqcup G_U$ and a $p\in \mathcal P_0(G_U)$ we see that all sites at the end of the trajectory must be of the same state $T\in \{S,I\}$ in configuration space.
This state, determined by the starting domain wall configuration and the partition, $T(G,p)\in \{I,S\}$, we refer to as the \emph{type} of the interval. Moreover, the type of an interval is determined by the boundaries in the configuration space. 
In other words, if $(x,y)\in \mathcal P_{0}(G_U)$, then $T(G,(x,y))=\vec \gamma^{(0)}_{x+1}=\vec\gamma^{(0)}_{y}$ where $\vec \gamma=\Gamma(G)$. 
Furthermore, once the type of one interval is determined, the others follow as they must alternate. Lastly, given an interval $(x,y)\in \mathcal P_0(G_U)$, we let $\Int(x,y)=\emptyset$ if $y=x+1$ and $\Int(x,y)=\{x+1,...,y-1\}$ be the \emph{interior points}. 

Since the $\ell$ surviving domain walls in $G_U$ partition the 1D circuit into $\ell$ intervals, these intervals are denoted by $p_j=(x_j,y_j)\in \mathcal P_0(G_U)$ for $j=1,...,\ell$, where $p_1=(x_1,y_1)$ satisfies $x_1=\min(G_U)$ and let $\tau_j$ denote the type of the interval $p_j$. Using these notions, the partition function can be expressed by decomposing it over disjoint intervals. Note that for $t\geq1$ the number of domain walls is at most $\frac{n}{2}$ and since there is always an even number of domain walls, we get:
\begin{align}Z^{(s)}&=3^k\cdot\frac{2^n}{3^n}\sum_{G\in \mathcal G_{\mathcal B}}\lambda^{|\Gamma(G)^{(s)}|}W(G)\\
&=3^k\cdot\frac{2^n}{3^n}\sum_{G_0\in \mathcal G_{\mathcal B,0}}\lambda^{|\Gamma(G_0)^{(s)}|}W(G_0)\label{eq:secondterm}\\
&\quad +3^k\cdot\frac{2^n}{3^n}\sum_{k_0=1}^{n/4}\sum_{\substack{G\in \mathcal G_{U,2k_0}\\ \exists \gamma\in \mathcal I:\DW(\gamma)=g^{(0)}}}W(G_U)\sum_{\substack{\tau_1\in \{I,S\}\\ \exists \gamma \in \mathcal I:\\(\gamma_{x_i+1})_{i=1}^{2k_0}=(\tau_i)_{i=1}^{2k_0}\\\andd (\gamma_{y_i})_{i=1}^{2k_0}=(\tau_i)_{i=1}^{2k_0}}}\prod_{j=1}^{2k_0}\sum_{\substack{H_{p_j}\in \mathcal G_0\\H_{p_j}\cap G_U=\emptyset\\ \exists \gamma\in \mathcal I:\\ h_{p_j}^{(0)}\subseteq \Int(p_j)\cap \DW(\gamma)\\ \andd \gamma_{x_j+1}=\gamma_{y_j}=\tau_j}}\lambda^{(S_{p_j}^{(s)}+1)\1_{\{\tau_j=S\}}}W(H_p).  \nonumber
\end{align}
The first term in Equation \eqref{eq:secondterm} constitutes the total contribution of all completely annihilating trajectories. In the second term, we sum over domain wall trajectories constructed as follows. The first sum iterates over the number of surviving domain walls, and the second over all domain wall trajectories with a fixed number of domain walls throughout.  However, we count only those that are commensurate with the initial conditions $\DW(\mathcal I)$. Next, a valid type is chosen for the interval $p_1$, which settles the types for all other intervals. A type selection is valid if there is a starting configuration $\gamma\in \mathcal I$ satisfying all the boundary conditions imposed by the types. Afterward, taking the product over all the intervals, we enforce the restraint that the domain wall trajectory is disjoint with the surviving domain walls. This ensures that the domain wall trajectories of the intervals are disjoint and are contained within the boundaries set by the surviving domain walls. Lastly, for the annihilating domain wall trajectories within the intervals, valid starting configurations compatible with the boundary conditions set by the type of the interval are imposed. \\

Since all annihilating trajectories in $s$ time steps are a subset of all annihilating trajectories within any finite number of time steps, we find that the first term in Equation \eqref{eq:secondterm} can be upper bounded by $Z^{(\infty)}$, which is determined in Lemma~\ref{lem:infinitedepth}. Thus, we focus our attention on the second term in Equation~\eqref{eq:secondterm}. Dropping the constraint that the type has to be valid globally on the second term in Equation~\eqref{eq:secondterm} immediately gives the following upper bound on the second term

\begin{align} \label{eq:secondterm2} 3^k\cdot\frac{2^n}{3^n}\sum_{k_0=1}^{n/4}\sum_{\substack{G_U\in \mathcal G_{U,2k_0}\\ \exists \gamma\in \mathcal I:\DW(\gamma)=g^{(0)}}}W(G_U)\sum_{\tau_1\in \{I,S\}}\overbrace{\prod_{j=1}^{2k_0}\sum_{\substack{H_{p_j}\in \mathcal G_0\\H_{p_j}\cap G_U=\emptyset\\ \exists \gamma\in \mathcal I:\\ h_{p_j}^{(0)}\subseteq \Int(p_j)\cap \DW(\gamma)\\ \andd \gamma_{x_j+1}=\gamma_{y_j}=\tau_j}}\lambda^{(S_{p_j}^{(s)}+1)\1_{\{\tau_j=S\}}}W(H_p)}^{\ast}. 
\end{align}
Our next goal is to decompose the expression under $\ast$ into \emph{non-interacting} intervals for a given $G_U$. The dependency lies in the requirement $H_{p_j}\cap G_U=\emptyset$, i.e., the trajectories arising from an interval $p_j$ can never cross the surviving domain walls at the boundary of the interval. To eliminate this dependency, we introduce three more notions.

\begin{definition}\label{def:orderpreserving}
    If $x=bq+r$ with $0\leq r <b$, we let $B(x)=q$ denote the block $x$ lies in. Furthermore, given an interval $(x,y)\in \mathcal P_0(G_U)$, we let $m_p$ be the number of complete $b$-sized blocks at time $t=0$. More precisely, if $x<y$ then $m_{(x,y)}=\max\{B(y)-B(x)-1,0\}$. Moreover, given a pair of domain walls at $t=0$, $(x_0,y_0)$, let $(x_t,y_t)_{t\in \{0,...,s\}}$ denote their evolution. 
\end{definition}

Now, we are ready to upper bound the partition function into factors defined in terms of non-interacting intervals. The sets in the following definition will play a significant role in this pursuit.

\begin{definition}
    Given an $(x,y)\in \mathcal P_0(G_U)$, where, if $\hat z = z\mod b$ for $z\in \mathbb N$, we define the following two sets when $B(x)\neq B(y)$:
    \begin{align}
    A_{(x,y)}^S&=\{S\}^{x+1} \times \{I,S\}^{b - \hat x - 1} \times (\{S\}^a \times \{I,S\}^{b-a})^{m_p}\times\{I,S\}^{\hat y} \times\{S\}^{n-(b-\hat x+x+bm_p+\hat y)}\\
    A_{(x,y)}^I&=\{I\}^{x+1} \times \{I,S\}^{b - \hat x - 1} \times (\{S\}^a \times \{I,S\}^{b-a})^{m_p}\times\{I,S\}^{\hat y} \times\{I\}^{n-(b-\hat x+x+bm_p+\hat y)}.
    \end{align}
If they do lie in the same block, i.e., $B(x)=B(y)$, we moreover define $A^T_{(x,y)}=\{T\}^{x+1}\times \{I,S\}^{y-x-1}\times \{T\}^{n-y}$ for $T\in \{I,S\}$.
\end{definition}

Note that the set $A^{\tau_j}_{p_j}$ is a superset of the set of all configurations $\gamma\in \mathcal I$ with the properties $h_{p_j}^{(0)}\subseteq \Int(p_j)\cap \DW(\gamma)$ and $\gamma_{x_j+1}=\gamma_{y_j}=\tau_j$. 
Thus, all the terms in the sum of $\gamma\in \mathcal I$ satisfying those two conditions are also counted if we sum over $A^{\tau_j}$ instead, giving us an upper bound. 
We are now ready to upper bound expression $\ast$ in Equation~\eqref{eq:secondterm2}.

The main idea of this upper bound is that by relaxing the conditions the sum is taken over, we obtain an upper bound since the sum will be taken over strictly more nonnegative terms.
We begin by adding domain wall trajectories to $H_{p_j}$ that are unrestricted by its boundary surviving domain walls. This means more trajectories are allowed as some might move past the boundary domain walls. Note that this "decouples" $H_{p_j}$ from $G_U$. 
However, some of these unrestricted trajectories might change the type of the interval as domain walls can now loop around (because of periodic boundary conditions) and annihilate with another domain wall which would change the type of the interval. For this reason, we must add only \emph{order-preserving} annihilations in $H_{p_j}$. 
For a pair of domain walls $(x_t,y_t)_{t\in \{0,...,s\}}$ that annihilate at time step $s$, assuming w.l.o.g. that $x_0<y_0$, we call the annihilation order-preserving if neither $(x_{s-1},y_{s-1})=(0,n-1)$ nor $x_{s-1}>y_{s-1}$ hold. Let $\mathcal O$ denote the set of domain wall trajectories with order-preserving annihilations. 

With this extension the configuration space of $H_{p_j}$ is separate from that of $G_U$ and the final configuration of $H_{p_j}$ will be $\{\tau_j\}^n$.
We can add even more trajectories by noting that $A^{\tau_j}$ is a superset of the set of initial conditions $\gamma\in \mathcal I$ with the properties $h_{p_j}^{(0)}\subseteq \Int(p_j)\cap \DW(\gamma)$ and $\gamma_{x_j+1}=\gamma_{y_j}=\tau_j$. Combining all steps and continuing from the sum in expression $\ast$ from Equation~\eqref{eq:secondterm2} implies for a given $G_U$ the following inequality
    \begin{equation}\label{eq:lemmareplacement1}
    \sum_{\substack{H_{p_j}\in \mathcal G_0\\H_{p_j}\cap G_U=\emptyset\\ \exists \gamma\in \mathcal I:\\ h_{p_j}^{(0)}\subseteq \Int(p_j)\cap \DW(\gamma)\\ \andd \gamma_{x_j+1}=\gamma_{y_j}=\tau_j}}\lambda^{(S_{p_j}^{(s)}+1)\1_{\{\tau_j=S\}}}W(H_p)\leq \sum_{\substack{H_{p_j}\in \mathcal G_0\cap \mathcal O\\ \exists \gamma\in \mathcal A^{\tau_j}_{p_j}:\DW(\gamma)=h_{p_j}^{(0)}}}\lambda^{(S_{p_j}^{(s)}+1)\1_{\{\tau_j=S\}}}W(H_p).
\end{equation}
Using that the configuration associated to $H_{p_j}$ converges to $\{\tau_j\}^n$ we obtain
\begin{equation} \label{eq:lemmareplacement2}
\sum_{\substack{H_{p_j}\in \mathcal G_0\cap \mathcal O\\ \exists \gamma\in \mathcal A^{\tau_j}_{p_j}:\DW(\gamma)=h_{p_j}^{(0)}}}\lambda^{(S_{p_j}^{(s)}+1)\1_{\{\tau_j=S\}}}W(H_p)
\leq \sum_{\substack{\vec \gamma\in \{I,S\}^{n\times (s+1)} \\ \vec\gamma^{(0)}\in A^{\tau_j}_{p_j} \\ \exists k\leq s\; :\; \vec \gamma^{(k)}= \{\tau_j\}^n}}\lambda^{(S_{p_j}^{(s)}+1)\1_{\{\tau_j=S\}}}\prod_{t=1}^sM^{(t)}_{\vec \gamma^{(t)}\vec \gamma^{(t-1)}},
\end{equation}
since any even number of non-order-preserving annihilations allows a trajectory starting from $A_{p_j}^{\tau_j}$ to converge to $\{\tau_j\}^n$. Summarizing the bounds in Equations~\eqref{eq:lemmareplacement1} and \eqref{eq:lemmareplacement2}, we have now the following upper bound on the second term in Equation~\eqref{eq:secondterm}:
\begin{align}
    &3^k\cdot\frac{2^n}{3^n}\sum_{k_0=1}^{n/4}\sum_{\substack{G_U\in \mathcal G_{U,2k_0}\\ \exists \gamma\in \mathcal I:\DW(\gamma)=g^{(0)}}}W(G_U)\sum_{\substack{\tau_1\in \{I,S\}\\ \exists \gamma \in \mathcal I:\\(\gamma_{x_i+1})_{i=1}^{2k_0}=(\tau_i)_{i=1}^{2k_0}\\\andd (\gamma_{y_i})_{i=1}^{2k_0}=(\tau_i)_{i=1}^{2k_0}}}\prod_{j=1}^{2k_0}\sum_{\substack{H_{p_j}\in \mathcal G_0\\H_{p_j}\cap G_U=\emptyset\\ \exists \gamma\in \mathcal I:\\ h_{p_j}^{(0)}\subseteq \Int(p_j)\cap \DW(\gamma)\\ \andd \gamma_{x_j+1}=\gamma_{y_j}=\tau_j}}\lambda^{(S_{p_j}^{(s)}+1)\1_{\{\tau_j=S\}}}W(H_p)\nonumber\\
    &\leq  3^k\cdot\frac{2^n}{3^n}\sum_{k_0=1}^{n/4}\sum_{\substack{G_U\in \mathcal G_{U,2k_0}\\ \exists \gamma\in \mathcal I:\DW(\gamma)=g^{(0)}}}W(G_U)\sum_{\tau_1\in \{I,S\}}\prod_{j=1}^{2k_0}\overbrace{\sum_{\substack{\vec \gamma\in \{I,S\}^{n\times (s+1)} \\ \vec\gamma^{(0)}\in A^{\tau_j}_{p_j} \\ \exists k\leq s\; :\; \vec \gamma^{(k)}= \{\tau_j\}^n}}\lambda^{(S_{p_j}^{(s)}+1)\1_{\{\tau_j=S\}}}\prod_{t=1}^sM^{(t)}_{\vec \gamma^{(t)}\vec \gamma^{(t-1)}}}^{\ast_2}\label{eq:secondterm3}
\end{align}

Since the intervals in Equation~\eqref{eq:secondterm3} are independent and all of the same form, we can derive one general bound depending on the type for the expression in $\ast_2$ to upper bound the term as a whole. Note that the expression in $\ast_2$ contains only trajectories that annihilate in $s$ steps to $\{\tau_j\}^n$. 
Following a similar computation as in Lemma~\ref{lem:infinitedepth}, this can be easily upper bounded by considering all trajectories that converge to $\{\tau_j\}^n$ at any depth. This gives the following lemma.
\begin{lemma}\label{lemma:ultimate}
     Let $p_j:=(x_j,y_j)\in \mathcal P_0(G_U)$, $m_{p_j}$ refer to the number of complete $b$-sized blocks in the interval at $t=0$ from Definition~\ref{def:orderpreserving} and assume w.l.o.g. that $0<x_j<y_j$. If $f\leq 1-\frac{a}{b}=1-\frac{k}{n}$, then for large enough $n$, the following holds:
    \begin{equation}\sum_{\substack{\vec \gamma\in \{I,S\}^{n\times (s+1)} \\ \vec\gamma^{(0)}\in A^{\tau_j}_{p_j} \\ \exists k\leq s\; :\; \vec \gamma^{(k)}= \{\tau_j\}^n}}\lambda^{(S_{p_j}^{(s)}+1)\1_{\{\tau_j=S\}}}\prod_{t=1}^sM^{(t)}_{\vec \gamma^{(t)}\vec \gamma^{(t-1)}}\leq \begin{cases}
        3^{2b}\cdot 2^{\frac{a}{b}S_{p_j}^{(0)}-\frac{a}{b}S_{p_j}^{(s)}}\cdot \frac{3^{(b-a)m_{p_j}}}{2^{bm_{p_j}}}, &\text{if $\tau_j=S$}\\
        3^{2b}\cdot \frac{3^{(b-a)m_{p_j}}}{2^{bm_{p_j}}}, &\text{if $\tau_j=I$}.
    \end{cases}\end{equation}
\end{lemma}
\begin{proof}
    Since the set of trajectories annihilating in $s$ or fewer time steps is a subset of the set of domain wall trajectories that annihilate at any finite number of time steps, we see that
    \begin{equation}
        \sum_{\substack{\vec \gamma\in \{I,S\}^{n\times (s+1)} \\ \vec\gamma^{(0)}\in A^{\tau_j}_{p_j} \\ \exists k\leq s\; :\; \vec \gamma^{(k)}= \{\tau_j\}^n}}\lambda^{(S_{p_j}^{(s)}+1)\1_{\{\tau_j=S\}}}\prod_{t=1}^sM^{(t)}_{\vec \gamma^{(t)}\vec \gamma^{(t-1)}},
    \end{equation}
    is increasing in circuit size $s$. Using this gives:
    \begin{equation}
        \sum_{\substack{\vec \gamma\in \{I,S\}^{n\times (s+1)} \\ \vec\gamma^{(0)}\in A^{\tau_j}_{p_j} \\ \exists k\leq s\; :\; \vec \gamma^{(k)}= \{\tau_j\}^n}}\lambda^{(S_{p_j}^{(s)}+1)\1_{\{\tau_j=S\}}}\prod_{t=1}^sM^{(t)}_{\vec \gamma^{(t)}\vec \gamma^{(t-1)}}\leq \sum_{\substack{\vec \gamma\in \{I,S\}^{n\times (s+1)} \\ \vec\gamma^{(0)}\in A^{\tau_j}_{p_j} \\ \exists k\in \mathbb N\; :\; \vec \gamma^{(k)}= \{\tau_j\}^n}}\lambda^{(S_{p_j}^{(s)}+1)\1_{\{\tau_j=S\}}}\prod_{t=1}^sM^{(t)}_{\vec \gamma^{(t)}\vec \gamma^{(t-1)}}\label{eq:secondterm4}.
    \end{equation}
    To determine the RHS of the inequality in Equation~\eqref{eq:secondterm4}, consider the partition function: 
    \begin{equation}Q^{\tau_j,(s)}_{p_j}= \sum_{\substack{\vec \gamma\in \{I,S\}^{n\times (s+1)} \\ \vec\gamma^{(0)}\in A^{\tau_j}_{p_j}}}\lambda^{(S_{p_j}^{(s)}+1)\1_{\{\tau_j=S\}}}\prod_{t=1}^sM^{(t)}_{\vec \gamma^{(t)}\vec \gamma^{(t-1)}}.\end{equation}
    If we only consider the trajectories that converge to $\{\tau_j\}^n$, then this quantity is equal to the RHS of the inequality in Equation~\eqref{eq:secondterm4}. Following the proof of Lemma~\ref{lem:partition} and the derivation of the partition function in \cite{Dalzell_2022}, we see that
    \begin{equation}Q_{p_j}^{\tau_j,(s)}=c^{\tau_j}_{p_j}\lambda^{(S_{p_j}^{(s)}+1)\1_{\{\tau_j=S\}}}\tr\left(\E_{U_S\sim \chi_{\brick}^{(D)}}\left[U_S^{\otimes 2}\tr_R\left((\rho_{SR}^{\tau_j,p_j})^{\otimes 2}S_R\right)(U_S^\dagger)^{\otimes 2}\right]|1^n\rangle\langle 1^n|^{\otimes 2}\right), \end{equation}
    where $\rho_{SR}^{\tau_j,p_j}=\otimes_{i=1}^n\rho_{SR\{i\}}^{\tau_j,p_j}$ with
    \begin{equation}
        \rho_{SR\{i\}}^{\tau_j,p_j}=\begin{cases}
            |\hat \phi\rangle\langle \hat \phi|_{S_iR_i}, &\text{if $A^{\tau_j}_{p_j,i}=\{S\}$} \\
            \frac12 I_{S_i}\otimes |0\rangle\langle 0|_{R_i}, &\text{if $A^{\tau_j}_{p_j,i}=\{I\}$}\\
            |00\rangle\langle 00|_{S_iR_i}, &\text{if $A^{\tau_j}_{p_j,i}=\{I,S\}$}.
        \end{cases}
    \end{equation}
    Here, we used that $A^{\tau_j}_{p_j}$ is defined as a product over $n$ factors, i.e., $A^{\tau_j}_{p_j}=\prod_{i=1}^nA^{\tau_j}_{p_j,i}$. Furthermore, we define $c_{p_j}^{\tau_j}=4^{N_S+N_I}\cdot 6^{N_{I,S}}=2^n\cdot 2^{N_S+N_I}\cdot 3^{N_{I,S}}$, where $N_S$ is the number of $\{S\}$ factors in $A^{\tau_j}_{p_j}$ and $N_{I}$, $N_{I,S}$ are defined similarly.

    Now, we determine $Q_{p_j}^{\tau_j,(\infty)}$ restricted such that all trajectories converge to $\{\tau_j\}^n$. In that case, we have upper bounded the RHS of Equation~\eqref{eq:secondterm4}. In infinite depth, we again have that $\chi_{\brick}$ becomes exactly the Haar random distribution. Thus, we have that $\E_{U_S\sim\Haar(2^n)}\left[U_S^{\otimes 2}\tr_R\left((\rho_{SR}^{\tau_j,p_j})^{\otimes 2}S_R\right)(U_S^\dagger)^{\otimes 2}\right]=\alpha I+\beta S$. However, since we are only interested in terms converging to $\{\tau_j\}^n$, we shall only take one of these terms into account. Computing $\alpha$ and $\beta$ gives
    \begin{equation}
        \alpha= \frac{2^{-N_S}-2^{-N_I-n}}{2^{2n}-1}, \qquad \beta = \frac{2^{-N_I}-2^{-N_S-n}}{2^{2n}-1}.
    \end{equation}
    Thus, if $\tau_j=I$, we obtain that 
    \begin{equation}
        \sum_{\substack{\vec \gamma\in \{I,S\}^{n\times (s+1)} \\ \vec\gamma^{(0)}\in A^{I}_{p_j} \\ \exists k\in \mathbb N\; :\; \vec \gamma^{(k)}= \{I\}^n}}\prod_{t=1}^sM^{(t)}_{\vec \gamma^{(t)}\vec \gamma^{(t-1)}}\leq c^{I}_{p_j}\alpha\leq 2^{-2n}\cdot 2^{-N_S}c_{p_j}^{\tau_j}=2^{-n+N_I}\cdot 3^{N_{I,S}}.
    \end{equation}
    Considering the definition of $A_{(x,y)}^I$ we have that when $B(x)\neq B(y)$ that $N_{I,S}\leq  (b-a)m_{p_j}+2b$ as well as $n-N_I=N_{S}+N_{I,S}\geq b\cdot m_{p_j} $. Using this then gives for $\tau_j=I$ the following inequality
    \begin{equation}
        \sum_{\substack{\vec \gamma\in \{I,S\}^{n\times (s+1)} \\ \vec\gamma^{(0)}\in A^{I}_{p_j} \\ \exists k\in \mathbb N\; :\; \vec \gamma^{(k)}= \{I\}^n}}\prod_{t=1}^sM^{(t)}_{\vec \gamma^{(t)}\vec \gamma^{(t-1)}}\leq 2^{-n+N_I}\cdot 3^{N_{I,S}}\leq 3^{2b}\cdot \frac{3^{(b-a)m_{p_j}}}{2^{bm_{p_j}}}.
    \end{equation}
    Under the mild assumption that there is more than one block, we have that when $p_j=(x,y)$ and $B(x)=B(y)$ that $A_{(x,y)}\subseteq A_{(x,y+b)}$ and the upper bound follows from the above. 
    
    When $\tau_j=S$ we obtain the following
    \begin{align}
        \sum_{\substack{\vec \gamma\in \{I,S\}^{n\times (s+1)} \\ \vec\gamma^{(0)}\in A^{S}_{p_j} \\ \exists k\in \mathbb N\; :\; \vec \gamma^{(k)}= \{S\}^n}}\lambda^{S_{p_j}^{(s)}+1}\prod_{t=1}^sM^{(t)}_{\vec \gamma^{(t)}\vec \gamma^{(t-1)}}&\leq c^{S}_{p_j}\beta \\[-7ex]
        &\leq  \frac{\lambda^{S_{p_j}^{(s)}+1}}{1-2^{-2n}}\cdot 2^{N_S-n}\cdot 3^{N_{I,S}}\\
        &\leq \lambda^{S_{p_j}^{(s)}}\cdot 2^{N_S-n}\cdot 3^{N_{I,S}}\\
        &\leq 2^{(f-(1-\frac{a}{b}))S_{p_j}^{(s)}}\cdot 2^{-\frac{a}{b}S_{p_j}^{(s)}}\cdot 2^{-(b-a)m_{p_j}}\cdot 3^{(b-a)m_{p_j}+2b}\\
        &\leq 2^{-\frac{a}{b}S_{p_j}^{(s)}}\cdot 2^{-(b-a)m_{p_j}}\cdot 3^{(b-a)m_{p_j}+2b}\\
        &= 2^{\frac{a}{b}S_{p_j}^{(0)}-\frac{a}{b}S_{p_j}^{(s)}}\cdot 2^{-(b-a)m_{p_j}-\frac{a}{b}S_{p_j}^{(0)}}\cdot 3^{(b-a)m_{p_j}+2b}\\
        &\leq 3^{2b}\cdot 2^{\frac{a}{b}S_{p_j}^{(0)}-\frac{a}{b}S_{p_j}^{(s)}}\cdot \frac{3^{(b-a)m_{p_j}}}{2^{bm_{p_j}}}.
    \end{align}
    In the third inequality we used that $f\leq 1-\frac{a}{b}<1$, meaning that for large enough $n$ that $\frac{\lambda}{1-2^{-2n}}\leq 1$. The fourth inequality is obtained by applying $n-N_{S}=N_{I,S}\geq (b-a)m_{p_j}$ and $N_{I,S}\leq (b-a)m_{p_j}+2b$. The fifth inequality is derived by using the assumption of the statement $f\leq 1-\frac{a}{b}$. The last inequality follows by using that $S_{p_j}^{(0)}\geq bm_{p_j}$. Lastly, the case $B(x)=B(y)$ and $\tau_j=S$ follows analogously since $A^S_{(x,y)}\subseteq A^S_{(x,y+b)}$ under the mild assumption that there is more than 1 block.
\end{proof}

With Lemma~\ref{lemma:ultimate} proven, we are ready to prove the main upper bound on the partition function.

\begin{theorem}\label{thm:mainresultthesis}
     Let $D=\frac{2s}{n}$ denote the depth of the random circuit, where $s$ is the number of 2-qubit gates in the brickwork architecture and $n$ the number of qubits. If $f\leq 1-\frac{a}{b}=1-R$ with $R=\frac{k}{n}$ and letting $C=10\cdot 3^{2b}2^b$ and $r=\frac{2(2^R+2^{-R})}{5}$, we have:
     \begin{equation}Z^{(s)}\leq Z^{(\infty)}\exp\left(nC\cdot r^D\right).\end{equation}
\end{theorem}
\begin{proof}
The main idea of this proof is to combine the upper bound in Equation~\eqref{eq:secondterm3} with Lemma~\ref{lemma:ultimate}. After this, we only have to upper bound the contribution of the surviving domain walls. To this end, we shall also assume that they are non-interacting as that will allow for more trajectories than when they are restricting each other.
Combining Equations~\eqref{eq:secondterm} and \eqref{eq:secondterm3} we obtain
{\allowdisplaybreaks
\begin{align}
    Z^{(s)}&\leq3^k\cdot\frac{2^n}{3^n}\sum_{G_0\in \mathcal G_{\mathcal B}}\lambda^{|\Gamma(G_0)^{(s)}|}W(G_0) \\
    &\space{3em} + 3^k\cdot\frac{2^n}{3^n}\sum_{k_0=1}^{n/4}\sum_{\substack{G_U\in \mathcal G_{U,2k_0}\\ \exists \gamma\in \mathcal I:\DW(\gamma)=g^{(0)}}}W(G_U)\sum_{\tau_1\in \{I,S\}}\prod_{j=1}^{2k_0}\sum_{\substack{\vec \gamma\in \{I,S\}^{n\times (s+1)} \\ \vec\gamma^{(0)}\in A^{\tau_j}_{p_j} \\ \exists k\leq s\; :\; \vec \gamma^{(k)}= \{\tau_j\}^n}}\lambda^{(S_{p_j}^{(s)}+1)\1_{\{\tau_j=S\}}}\prod_{t=1}^sM^{(t)}_{\vec \gamma^{(t)}\vec \gamma^{(t-1)}}\nonumber\\
    &\leq Z^{(\infty)} + 3^k\cdot\frac{2^n}{3^n}\sum_{k_0=1}^{n/4}\sum_{\substack{G_U\in \mathcal G_{U,2k_0}\\ \exists \gamma\in \mathcal I:\DW(\gamma)=g^{(0)}}}W(G_U)\sum_{\tau_1\in \{I,S\}}\prod_{j=1}^{2k_0}\sum_{\substack{\vec \gamma\in \{I,S\}^{n\times (s+1)} \\ \vec\gamma^{(0)}\in A^{\tau_j}_{p_j} \\ \exists k\leq s\; :\; \vec \gamma^{(k)}= \{\tau_j\}^n}}\lambda^{(S_{p_j}^{(s)}+1)\1_{\{\tau_j=S\}}}\prod_{t=1}^sM^{(t)}_{\vec \gamma^{(t)}\vec \gamma^{(t-1)}}.
\end{align}}
Applying Lemma~\ref{lemma:ultimate} gives:
{\allowdisplaybreaks
\begin{align}
    \nonumber Z^{(s)}&\leq Z^{(\infty)}+3^k\cdot\frac{2^n}{3^n}\sum_{k_0=1}^{n/4}\sum_{\substack{G_U\in \mathcal G_{U,2k_0}\\ \exists \gamma\in \mathcal I:\DW(\gamma)=g^{(0)}}}W(G_U)\sum_{\tau_1\in \{I,S\}}\prod_{j=1}^{2k_0}\sum_{\substack{\vec \gamma\in \{I,S\}^{n\times (s+1)} \\ \vec\gamma^{(0)}\in A^{\tau_j}_{p_j} \\ \exists k\leq s\; :\; \vec \gamma^{(k)}= \{\tau_j\}^n}}\lambda^{(S_{p_j}^{(s)}+1)\1_{\{\tau_j=S\}}}\prod_{t=1}^sM^{(t)}_{\vec \gamma^{(t)}\vec \gamma^{(t-1)}}\\
    &\leq Z^{(\infty)}+3^k\cdot\frac{2^n}{3^n}\sum_{k_0=1}^{n/4}\sum_{\substack{G_U\in \mathcal G_{U,2k_0}\\ \exists \gamma\in \mathcal I:\DW(\gamma)=g^{(0)}}}W(G_U)\sum_{\tau_1\in \{I,S\}}\prod_{j=1}^{2k_0}3^{2b}\cdot 2^{\left(\frac{a}{b}S_{p_j}^{(0)}-\frac{a}{b}S_{p_j}^{(s)}\right)\1_{\{\tau_j=S\}}}\cdot \frac{3^{(b-a)m_{p_j}}}{2^{bm_{p_j}}}. \label{eq:secondterm5}
\end{align}}Simplifying the product in Equation~\eqref{eq:secondterm5}, by using that $m-k_0\geq \sum_{j=1}^{2k_0}m_{p_j}\geq m-2k_0$, gives:
\begin{equation}
    \prod_{j=1}^{2k_0}  \frac{3^{(b-a)m_{p_j}}}{2^{bm_{p_j}}}\leq \frac{3^{(b-a)(m-k_0)}}{2^{b(m-2k_0)}}=\frac{3^{-(b-a)k_0}}{2^{-2bk_0}}\cdot\frac{3^{bm}}{2^{bm}}\cdot 3^{-am}\leq 2^{2bk_0}\cdot \frac{3^n}{2^n}\cdot 3^{-k}.
\end{equation}
Plugging this in gives:
\begin{align}
    Z^{(s)}&\leq Z^{(\infty)}+\sum_{k_0=1}^{n/4}3^{4bk_0}2^{2bk_0}\sum_{\substack{G_U\in \mathcal G_{U,2k_0}\\ \exists \gamma\in \mathcal I:\DW(\gamma)=g^{(0)}}}W(G_U)\sum_{\tau_1\in \{I,S\}}\prod_{j=1}^{2k_0}2^{\left(\frac{a}{b}S_{p_j}^{(0)}-\frac{a}{b}S_{p_j}^{(s)}\right)\1_{\{\tau_j=S\}}}\\
    &\leq Z^{(\infty)}+\sum_{k_0=1}^{n/4}3^{4bk_0}2^{2bk_0}\sum_{\substack{G_U\in \mathcal G_{U,2k_0}}}W(G_U)\sum_{\tau_1\in \{I,S\}}\prod_{j=1}^{2k_0}2^{\left(\frac{a}{b}S_{p_j}^{(0)}-\frac{a}{b}S_{p_j}^{(s)}\right)\1_{\{\tau_j=S\}}},
\end{align}
where in the last inequality, we omitted the constraint that for $G_U\in \mathcal G_{U,2k_0}$ that there must be some $\gamma\in \mathcal I$ such that $\DW(\gamma)=g^{(0)}$. This gives an upper bound since we sum over more starting states and will make the counting later on easier.

Each $S$-interval must have exactly 2 surviving boundary domain walls and these pairs partition the surviving domain walls. Partitioning the surviving domain walls in these pairs, we have that $G_U=\bigsqcup_{j=1}^{2k_0}G_{p_j}$ where $G_{p_j}\in \mathcal G_{U,2}$. Let $p_j=(x_{j0},y_{j0})$ and assume w.l.o.g. that $x_{j0}<y_{j0}$. In the first layer, they might not move. After the first layer though, they will both move either left or right at every layer. With this in mind, we see that 
\begin{equation}
    W(G_{p_j})\leq \left(\frac{2}{5}\right)^{2D-2}.
\end{equation}
So, we shall focus on $2^{\frac{a}{b}S_{p_j}^{(0)}-\frac{a}{b}S_{p_j}^{(s)}}$. Letting $G_{p_j}=(g_{p_j}^{(t)})_{t\in \{0,...,D\}}:=(x_{jt},y_{jt})_{t\in \{0,...,D\}}$ (note that this is not exactly a domain wall trajectory, but it can be identified to one in a bijective manner) denote the evolution of these surviving domain walls, we see that
\begin{equation}2^{\frac{a}{b}S_{p_j}^{(0)}-\frac{a}{b}S_{p_j}^{(s)}}\leq 9\cdot 2^{2\frac{a}{b}}\prod_{t=2}^s2^{\Delta_t(G_{p_j})},\end{equation}
where, in the worst case, at $t=0$ to $t=1$, both domain walls move inwards and there are 9 possible movements on the first time step , giving the factor of $9\cdot 2^{2\frac{a}{b}}$ and for $t>1$ we define:
\begin{equation}\Delta_t(G_{p_j})=\begin{cases} 2^{-2\frac{a}{b}}, &\text{if $(x_{jt},y_{jt})=(x_{j(t-1)}-1,y_{j(t-1)}+1)$}\\
1, &\text{if $(x_{jt},y_{jt})=(x_{j(t-1)}+1,y_{j(t-1)}+1)$ or $(x_{jt},y_{jt})=(x_{j(t-1)}-1,y_{j(t-1)}-1)$}\\
2^{2\frac{a}{b}}, &\text{if $(x_{jt},y_{jt})=(x_{j(t-1)}+1,y_{j(t-1)}-1)$}.
\end{cases}\end{equation}
Now, we are interested in the sum of this over all allowed walks, i.e., for $p_j=(x_{j0},y_{j0})$ (fixed), we want to consider:
\begin{equation}\sum_{\substack{G_{p_j}\in \mathcal G_{U,2}\\g_{p_j}^{(0)}=\{x_{j0},y_{j0}\}}}2^{\frac{a}{b}S_{p_j}^{(0)}-\frac{a}{b}S_{p_j}^{(s)}}.\end{equation}
However, if we let $\mathcal W$ be the set of all walks, where after the first layer both domain walls can go either left or right, are allowed to \emph{cross each other}, and use $w=(x_{jt},y_{jt})_{t\in \{0,...,D\}}$ to denote such a walk we get:
\begin{equation}\sum_{\substack{G_{p_j}\in \mathcal G_{U,2}\\g_{p_j}^{(0)}=\{x_{j0},y_{j0}\}}}2^{\frac{a}{b}S_p^0-\frac{a}{b}S_p^f}=\sum_{\substack{G_{p_j}\in \mathcal G_{U,2}\\g_{p_j}^{(0)}=\{x_{j0},y_{j0}\}}}9\cdot 2^{2\frac{a}{b}}\prod_{t=2}^s2^{\Delta_t(G_{p_j})}\leq 9\cdot 2^{2\frac{a}{b}}\sum_{w\in \mathcal W}\prod_{t=2}^s2^{\Delta_t(w)},\end{equation}
as $\mathcal G_{U,2}$ can be interpreted as a subset of $\mathcal W$. Hence, using this upper bound and summing over all possible paths gives:
\begin{equation}\sum_{\substack{G_{p_j}\in \mathcal G_{U,2}\\g_{p_j}^{(0)}=\{x_{j0},y_{j0}\}}}2^{\frac{a}{b}S_{p_j}^{(0)}-\frac{a}{b}S_{p_j}^{(s)}}\leq 9\cdot 2^{2\frac{a}{b}}\left(2^{2\frac{a}{b}}+2+2^{-2\frac{a}{b}}\right)^{D-1}.\end{equation}
Substituting $R=\frac{a}{b}$ and
combining this with the previous factor of $\left(\frac{2}{5}\right)^{2D-2}$, we obtain:
\begin{align}\sum_{\substack{G_{p_j}\in \mathcal G_{U,2}\\g_{p_j}^{(0)}=\{x_{j0},y_{j0}\}}}W(G_{p_j})2^{\frac{a}{b}S_{p_j}^{(0)}-\frac{a}{b}S_{p_j}^{(s)}}&\leq 9\cdot 2^{2R}\cdot \frac{25}{9}\cdot \frac{1}{2^{2R}+2+2^{-2R}}\left(2^{2R}+2+2^{-2R}\right)^{D}\cdot\left(\frac{4}{25}\right)^D\\
&\leq 25 \left(\frac{2\sqrt{2^{2R}+2+2^{-2R}}}{5}\right)^{2D}=5^2\left(\frac{2(2^R+2^{-R})}{5}\right)^{2D}.\label{eq:UBwalk}
\end{align}
Letting $r=\frac{2(2^R+2^{-R})}{5}$ we note that for $R<1$ that $r<1$.
The number of possible partitions given $k_0$ pairs of surviving domain walls is upper bounded by $\binom{n}{2k_0}$. Combining all of this, we have:
\begin{align}
    Z^{(s)}&\leq  Z^{(\infty)}+\sum_{k_0=1}^{n/4}3^{4bk_0}2^{2bk_0}\sum_{\substack{G_U\in \mathcal G_{U,2k_0}}}W(G_U)\sum_{\tau_1\in \{I,S\}}\prod_{j=1}^{2k_0}2^{\left(\frac{a}{b}S_{p_j}^{(0)}-\frac{a}{b}S_{p_j}^{(s)}\right)\1_{\{\tau_j=S\}}}\\
    &\leq Z^{(\infty)}+2\sum_{k_0=1}^{n/4}\binom{n}{2k_0}(5\cdot 3^{2b}2^{b})^{2k_0}r^{2Dk_0} \label{eq:factor2}\\
    &\leq Z^{(\infty)}+\sum_{k_0=1}^{n/4}\binom{n}{2k_0}(10\cdot 3^{2b}2^{b})^{2k_0}r^{2Dk_0}, \label{eq:QECsecondtermbound}
\end{align}
where $C=10\cdot 3^{2b}2^{b}$. The factor of $2$ in Equation~\eqref{eq:factor2} comes from summing over both types $\tau_1\in \{I,S\}$, but noting that the upper bound derived in Equation~\eqref{eq:UBwalk} is independent of the set of types. Using that $Z^{(\infty)}\geq 1$, we obtain that
\begin{equation}Z^{(s)}\leq Z^{(\infty)} \sum_{k_0=0}^{n/4}\binom{n}{2k_0}C^{2k_0}r^{2Dk_0}\leq Z^{(\infty)}\left(1+C\cdot r^D\right)^n \leq Z^{(\infty)}\exp\left(nC\cdot r^D\right),\end{equation}
as desired.
\end{proof}

Now, it becomes clear when $f<1-\frac{a}{b}$ that Theorem~\ref{thm:informal} can be obtained by applying Theorem~\ref{thm:mainresultthesis} to Equation~\eqref{eq:expectationterm} and noting that $Z^{(\infty})\exp(\cdot) -1 = (Z^{(\infty)}-1)+Z^{\infty}(\exp(\cdot)-1)$. Here, by Lemma~\ref{lem:infinitedepth}, $ (Z^{(\infty)}-1)$ goes to zero exponentially fast. Furthermore, note that $Z^{\infty}(\exp(nCr^{D(n)})-1)\in \mathcal O(nr^{D(n)})$ from which the bound in Theorem~\ref{thm:informal} follows by choosing $D(n)=c\log(n)$ with $c>\frac{1}{|\log r|}$.

\section{Exact Error Correction}\label{sec:EC}
In this section we cover the exact error correcting properties of the random Clifford brickwork circuit. We will denote the Hermitian single-qubit Pauli operators as $\sigma_\nu$ with $\nu\in \{I,X,Y,Z\}$. For an $n$ qubit Pauli operator $\sigma_\mu = \bigotimes_{i=1}^n \sigma_{\mu_i}$ we will denote by $\omega(\mu)$ the number of non-identity Pauli operators in $\sigma_\mu$. With this we follow the definition of an exact error correcting code in \cite{Brown}:

\begin{definition}\label{def:EC}
    A \emph{unitary encoder} encoding $k$-qubits into $n$-qubits is a unitary $U\in \U(2^n)$ that sends a $k$ qubit logical state $\ket{\psi}$ to $|\psi\rangle \mapsto U|\psi\rangle \otimes |0^{n-k}\rangle$. Letting $\{|x\rangle\in \mathbb C^{2^k}\mid x\in \{0,1\}^k\}$ denote the basis on the space of $k$-qubits, we define the subspace spanned by $|\bar x\rangle =U|x\rangle_L|0^{n-k}\rangle_{S\setminus L}$ of $\mathbb C^{2^n}$ to be the \emph{code space}, denoted by $\mathcal C$. We say that a code has \emph{distance} at least $d+1$ if for any $x,y\in \{0,1\}^k$ and any $\mu\in \{I,X,Y,Z\}^n$ with $1\leq \omega(\mu)\leq d$ we have:
    \begin{equation}\label{eq:EC}\langle \bar x |\sigma_{\mu}|\bar y\rangle=C_{\mu}\delta_{xy},\end{equation}
    for some $C_{\mu}\in \mathbb R$. If Equation~\eqref{eq:EC} is satisfied, we call $U$ a $[[n,k,d+1]]$ quantum error correcting code.
\end{definition}

The next statement is a useful \emph{equivalent} error correction criterion when the unitary encoder is an $n$-qubit Clifford gate. The set of $n$-qubit Clifford gates is denoted by $\mathcal C_n$.

\begin{lemma}[Proposition 2.1 in \cite{Brown}]\label{lemma:EC}
            Let $U\in \C_n$. Then, $U$ is an $[[n,k,d+1]]$ error correcting code if and only if for any $\nu_A\in \{I,X,Y,Z\}^k-\{0\}$ and $\nu_B\in \{I,Z\}^{n-k}$ and all $\mu\in \{I,X,Y,Z\}^n$ satisfying $1\leq \omega(\mu)\leq d$ we have:
    \begin{equation}\tr\left(\sigma_\mu U((\sigma_{\nu_A})_L\otimes (\sigma_{\nu_B})_{S\setminus L})U^\dagger\right)=0.\end{equation}
\end{lemma}
Consider $U\sim \chi_{\brick}^{(D)}$, then the probability that $U$ fails to satisfy the criterion in \eqref{eq:EC}, call it $\mathbb P(F)$, is equal to the probability that it fails to satisfy the criterion in Lemma \ref{lemma:EC}. Thus, using the union bound, we find that the probability of failure is bounded by
    \begin{align}\mathbb P(F)&=\mathbb P\left(\exists \nu_A,\nu_B,\mu \; \text{s.t.} \; 1\leq \omega(\mu)\leq d \andd \tr\left(\sigma_\mu U((\sigma_{\nu_A})_L\otimes (\sigma_{\nu_B})_{S\setminus L})U^\dagger\right)\neq 0\right)\\
    &\leq \sum_{\substack{\nu_A\in \{I,X,Y,Z\}^{k}\\\nu_A\neq I^{\times k}}}\sum_{\nu_B\in \{I,Z\}^{n-k}}\sum_{\substack{\mu\in \{I,X,Y,Z\}^n\\ 1\leq \omega(\mu)\leq d}}\mathbb P\left(\tr\left(\sigma_\mu U((\sigma_{\nu_A})_L\otimes (\sigma_{\nu_B})_{S\setminus L})U^\dagger\right)\neq 0\right).\label{eq:failureprob}
    \end{align}
From the expression in Equation~\eqref{eq:failureprob} we can derive the partition function through a similar calculation as for the proof of Lemma~\ref{lem:partition}. This derivation can be found in the appendix.
\begin{lemma}\label{lem:partition-EC}
Let the notation of Lemma~\ref{lem:partition} hold and $\mathbb P(F)$ be the probability that any $U\sim \chi_{\brick}$ is not a $[[n,k,d+1]]$ error correcting code. Then, assuming the block architecture for the qubits from Section~\ref{sec:brickworkcircuits}, the following bound holds:
    \begin{equation}
        \mathbb P(F)\leq Z^{(s)}_{\QEC}:=3^k\cdot \frac{2^n}{3^n}\sum_{\substack{\vec \gamma\in \{I,S\}^{n\times (s+1)}\\ \vec \gamma^{(0)}\in (\{S\}^a\times\{I,S\}^{b-a})^m}}\left(2^{-|\vec \gamma^{(s)}|}\sum_{j=1}^{\min\{d,|\vec \gamma^{(s)}|\}}3^j\binom{|\vec \gamma^{(s)}|}{j}\right)\prod_{t=1}^sM^{(t)}_{\vec \gamma^{(t)}\vec \gamma^{(t-1)}}.
\end{equation}
\end{lemma}

To avoid confusion we shall refer to the partition function in Lemma~\ref{lem:partition} as $Z_{\AQEC}^{(s)}$. Following a similar style of argument as in the proof of Lemma~\ref{lem:infinitedepth}, we can compute $Z^{(\infty)}_{\QEC}$ in this case as well. The result is stated in the following lemma, whose proof is deferred to the appendix.
\begin{lemma}\label{lem:infinitedepthexact}
    At infinite depth, the following equality holds
    \begin{equation}
        Z^{(\infty)}_{\QEC}=\frac{K}{1-2^{-2n}}\left(2^{-n+k}-2^{-2n}\right),
    \end{equation}
    where $K=\sum_{j=1}^d3^j\binom{n}{j}$.
\end{lemma}
We will now bound $Z_{\AQEC}^{(s)}$ by relating it to both $Z_{\AQEC}^{(\infty)}$ and $Z_{\QEC}^{(s)}$. For the latter connection, note that $2^{-|\vec \gamma^{(s)}|}=\lambda^{|\vec \gamma^{(s)}|}\cdot 2^{-f|\vec \gamma^{(s)}|}$, where we let $f=1-\frac{a}{b}$ (recall that $\lambda = 2^{f-1}$). By definition of $f$, we see that $0<f<1$ as $0<a<b$. So, if we upper bound $2^{-f|\vec \gamma^{(s)}|}\sum_{j=1}^{\min\{d,|\vec \gamma^{(s)}|\}}3^j\binom{|\vec \gamma^{(s)}|}{j}$ for all $1\leq |\vec \gamma^{(s)}|\leq n$ we have a direct relation between $Z_{\QEC}^{(s)}$ and $Z_{\AQEC}^{(s)}$. This brings us to the following lemma.

\begin{lemma}\label{lemma:globalboundEC} 
    Let $0<f<1$ and $A=\max\left\{2^{2-f},\frac{3}{f\ln(2)}\right\}$. Then for any $1\leq |\vec \gamma^{(s)}|\leq n$ the following bound holds:
    \begin{equation}
        2^{-f|\vec \gamma^{(s)}|}\sum_{j=1}^{\min\{d,|\vec \gamma^{(s)}|\}}3^j\binom{|\vec \gamma^{(s)}|}{j}\leq d\cdot A^d.
    \end{equation}
\end{lemma}

With Lemmas~\ref{lem:infinitedepthexact} and \ref{lemma:globalboundEC} in place we are ready to upper bound $Z_{\QEC}^{(s)}$.
\begin{theorem}\label{thm:mainresultEC}
    Assume the notation of Theorem~\ref{thm:mainresultthesis}. More precisely, let $C=10\cdot 3^{2b}\cdot 2^b$, $r=\frac{2(2^R+2^{-R}}{5}$ where $R=\frac{a}{b}=\frac{k}{n}$, and $A=\max\left\{2^{2-f},\frac{3}{f\ln(2)}\right\}$. Then the following bound holds:
    \begin{equation}
        Z^{(s)}_{\QEC}\leq  Z^{(\infty)}_{\QEC}+\exp\left(\ln(d)+d\ln(A)+\ln(n)+\ln(C)+D\ln(r)+nC\cdot r^D\right).
    \end{equation}
\end{theorem}
\begin{proof}
    Following the line of thought in Equation~\eqref{eq:secondterm}, we split $Z^{(s)}_{\QEC}$ up into domain wall trajectories that terminate and domain wall trajectories that do not. This gives the following equality:
\begin{align}\label{eq:QECsecondterm}
    Z^{(s)}_{\QEC}&= 3^k\frac{2^n}{3^n}\sum_{G_0\in \mathcal G_{\mathcal B,0}}\left(2^{-|\Gamma(G_0)^{(s)}|}\sum_{j=1}^{\min\{d,|\Gamma(G_0)^{(s)}|\}}3^j\binom{|\Gamma(G_0)^{(s)}|}{j}\right)W(G_0)\\
    &\qquad +3^k\frac{2^n}{3^n}\sum_{k_0=1}^{n/4}\sum_{G_U\in \mathcal G_{U,2k_0}}\sum_{\substack{H\in \mathcal G_0\\ G_U\cap H=\emptyset \\ g^{(0)}\sqcup h^{(0)}\in \mathcal B}}\left(2^{-|\Gamma(G_U\sqcup H)^{(s)}|}\sum_{j=1}^{\min\{d,|\Gamma(G_U\sqcup H)^{(s)}|\}}3^j\binom{|\Gamma(G_U\sqcup H)^{(s)}|}{j}\right)W(G_U\sqcup H)\nonumber.
\end{align}
Now, the first term in Equation~\eqref{eq:QECsecondterm} is easily seen to be upper bounded by $Z_{\QEC}^{(\infty)}$ as all the trajectories that terminate in $s$ steps are a subset of all the trajectories that terminate at all (and the weights are still nonnegative). For the second term, we shall use that $2^{-|\vec \gamma^{(s)}|}=\lambda^{|\vec \gamma^{(s)}|}\cdot 2^{-f|\vec \gamma^{(s)}|}$, where we let $f=1-\frac{a}{b}$. By definition of $f$, we see that $0<f<1$ as $0<a<b$. So, we can apply Lemma~\ref{lemma:globalboundEC} to the second term. Combining both these steps gives the inequality
\begin{equation}
    Z^{(s)}_{\QEC}\leq Z_{\QEC}^{(\infty)}+d\cdot A^d\cdot 3^k\frac{2^n}{3^n}\sum_{k_0=1}^{n/4}\sum_{G_U\in \mathcal G_{U,2k_0}}\sum_{\substack{H\in \mathcal G_0\\ G_U\cap H=\emptyset \\ g^{(0)}\sqcup h^{(0)}\in \mathcal B}}\lambda^{|\Gamma(G_U\sqcup H)^{(s)}|}W(G_U\sqcup H),
\end{equation}
where the second term is equivalent to the second term in Equation~\eqref{eq:secondterm} (which is the core quantity we analyzed when bounding $Z_{\AQEC}^{(s)}$) with $f=1-\frac{a}{b}$. Without the prefactor $d\cdot A^d$, we have bounded this term in the proof of Theorem~\ref{thm:mainresultthesis} in Equation~\eqref{eq:QECsecondtermbound}. So, this immediately gives
\begin{align}Z_{\QEC}^{(s)}&\leq Z^{(\infty)}_{\QEC}+d\cdot A^d\sum_{k_0=1}^{n/4}\binom{n}{2k_0}C^{2k_0}r^{2Dk_0} \\
&\leq Z^{(\infty)}_{\QEC}+d\cdot A^d\sum_{k_0=1}^{n}\binom{n}{k_0}C^{k_0}r^{Dk_0}\\
&=Z^{(\infty)}_{\QEC}+d\cdot A^d\left(\left(1+C\cdot r^D\right)^n-1\right)\\
&\leq Z^{(\infty)}_{\QEC}+d\cdot A^d\left(\exp\left(nC\cdot r^D\right)-1\right)\\
&\leq Z^{(\infty)}_{\QEC}+d\cdot A^d\left(nC\cdot r^D\right)\exp\left(nC\cdot r^D\right)\\
&=Z^{(\infty)}_{\QEC}+\exp\left(\ln(d)+d\ln(A)+\ln(n)+\ln(C)+D\ln(r)+nC\cdot r^D\right),
\end{align}
where in the second last inequality, we used that $\exp(z)-1\leq z\exp(z)$.
\end{proof}

Using Theorem~\ref{thm:mainresultEC} one immediately obtains the following corollary showing that a code distance of $d(n)$ can be achieved in depth $\mathcal O(d(n))$, effectively matching the lower bound obtained by a light-cone argument.

\begin{corollary} \label{cor:exacterrorcorrection}
    Assume the notation of Theorem~\ref{thm:mainresultEC} and that $d(n)\in \Omega(\log(n))$. Then, there exists $D(n)\in \mathcal O(d(n))$ such that the probability that $U\sim \chi_{\brick}^{D(n)}$ is a $[[n,k,d(n)+1]]$-code goes to 1 in the limit $n\to \infty$ for appropriate $k$. More specifically, if $d(n)\in o(n)$, the result holds when $k=\frac{a}{b}n$ with any $0<a<b$. When $d(n)$ is asymptotically linear in $n$ where $d(n)\leq c\cdot n$ for some $c\in (0,1)$, then the only constraint on the rate is given by $c\log(3)+H(c)+\frac{a}{b}<1$ where $H(\cdot)$ is the binary entropy function.
\end{corollary}
\begin{proof}
    Let us first analyze the expression $Z^{(\infty)}_{\QEC}$. Using the quantum singleton bound, we have that $d\leq \frac{n}{2}$ which gives
    \begin{equation}K=\sum_{j=1}^d3^j\binom{n}{j}\leq d\cdot 3^d\cdot \binom{n}{d}\leq d\cdot 3^d\cdot 2^{nH(d/n)}=2^{\log(d)+d\log(3)+nH(d/n)},\end{equation}
where we used $\binom{m}{\ell}\leq 2^{mH(\ell/m)}$. Combining this with Lemma~\ref{lem:infinitedepthexact} gives the following bound:
\begin{equation}
    Z^{(\infty)}_{\QEC}\leq 2^{\log(d)+d\log(3)+nH(d/n)+k-n}.
\end{equation}
Now, using that $k$ is linear in $n$ with $n=\frac{b}{a}k$ and assuming that there exists an $n_1$ such that $d\leq c\cdot n$ when $n>n_1$ for some $0<c<1$ yields the following:
\begin{equation}
    Z^{(\infty)}_{\QEC}\leq 2^{\log(cn)+cn\log(3)+nH(c)+\frac{a}{b}n-n},
\end{equation}
showing when $c\log(3)+H(c)+\frac{a}{b}<1$ that $Z^{(\infty)}_{\QEC}\to 0$ as $n \to \infty$. Note that when $d(n)$ lies in $o(n)$, $c$ can be taken to be arbitrarily small for large enough $n_1$ and $Z^{(\infty)}_{\QEC}$ will go to zero for any combination of $a$ and $b$ as long as $0<a<b$. Hence, the constraint $c\log(3)+H(c)+\frac{a}{b}<1$ is only necessary when $d(n)$ is asymptotically linear in $n$. 

Focusing on the second term in the bound of Theorem~\ref{thm:mainresultEC} we see that since $d(n)\in \Omega(\log(n))$, we can choose $D(n)=\alpha \cdot d(n)$ for some $\alpha>0$ such that the second term also goes to zero in the limit $n\to \infty$.
\end{proof}
The result in Theorem~\ref{thm:informal2} now follows by Corollary~\ref{cor:exacterrorcorrection}.

\begin{acknowledgements}
The bulk of this work was undertaken in the context of TK's 2025 master thesis project at the University of Amsterdam. While finishing up this manuscript we became aware of an independent derivation of our result by the authors of \cite{liu2025approximatequantumerrorcorrection}, which has recently appeared in the  journal version of this article (as well as the arXiv version-two). JH acknowledges funding from the Dutch Research Council (NWO) through a Veni grant (grant No.VI.Veni.222.331) and the Quantum Software Consortium (NWO Zwaartekracht Grant No.024.003.037).
\end{acknowledgements}
\bibliography{library}
\bibliographystyle{abbrvnat}
\appendix
\onecolumngrid
\section{Additional technical material}

\begin{proof}[Proof of Lemma~\ref{lem:partition}] The start of this proof mirrors the computation of Theorem 1 in \cite{liu2025approximatequantumerrorcorrection}. We apply the swap trick to the square in the trace, by using the cyclicity of the trace and the adjoint channel we get two factors that resemble second moments. Using Lemma~\ref{lemma:haartwirl} we compute the twirl of the single-qubit gates. To determine the effect of the 2-qubit gates, a similar approach will be taken as in Section B1 of \cite{Dalzell_2022}. Once all the single-qubit gates are applied we are left with a sum in which each term is a tensor product consisting of $n$ factors each being either $I$ or $S$. The 2-qubit gates will then move the positions of these $I$ and $S$ factors in each term. The way these factors move is encoded by $M^{(t)}_{\nu\gamma}$ from Lemma~\ref{lem:partition}.
Throughout this proof, $S$ is used to denote the swap operator, but also the physical system. The notation $S_A$ denotes that the swap operator acts on the system $A$. 
Using the swap trick and the adjoint channel $\tilde{\mathcal N}^\dagger$, one obtains
\begin{align}
    \tr\left(\left(\tilde{\mathcal N}_{S\to E}(U_S\rho_{SR}U_S^\dagger)\right)^2\right)&=\tr\left(\left(\tilde{\mathcal N}_{S\to E}(U_S\rho_{SR}U_S^\dagger)\right)^{\otimes 2}(S_E\otimes S_R)\right)\\
    &=\tr\left(U_S^{\otimes 2}\rho_{SR}^{\otimes 2}(U_S^\dagger)^{\otimes 2}((\tilde{\mathcal N}_{S\to E}^{\dagger})^{\otimes 2}(S_E)\otimes S_R)\right)\\
    &=\tr_S\left(U_S^{\otimes 2} \tr_R\left(\rho_{SR}^{\otimes 2}S_R\right)(U_S^\dagger)^{\otimes 2}(\tilde{\mathcal N}_{S\to E}^{\dagger})^{\otimes 2}(S_E)\right).
\end{align}

Next, suppose that to our architecture we add a layer of Haar random single-qubit gates at the beginning and the end. Then, this distribution of circuits is equivalent to $\chi_{\brick}$ because the gates in $\chi_{\brick}$ are sampled from unitary $2$-designs meaning that the layers of single-qubit gates can be absorbed into the adjacent layer of $2$-qubit gates by the invariance of the Haar measure.
Thus, w.l.o.g., we assume that $U_S$ is of the form
\begin{equation}\label{eq:RQCexplicitform2}
    U_S=V_{\{1\}}^{(-1)}V_{\{2\}}^{(-2)}...V_{\{n\}}^{(-n)}U^{(s)}_{A_s}...U^{(1)}_{A_1}U^{(-1)}_{\{1\}}...U^{(-n)}_{\{n\}}.
\end{equation}

Let $\chi$ denote the distribution over circuits of the form in Equation \eqref{eq:RQCexplicitform2}. Using the form of $U_S$ given in Equation \eqref{eq:RQCexplicitform2} and the cyclicity of the trace, we get
\begin{align}
    \tr\left(\left(\tilde{\mathcal N}_{S\to E}(U_S\rho_{SR}U_S^\dagger)\right)^2\right)&=\tr_S\left(U_S^{\otimes 2} \tr_R\left(\rho_{SR}^{\otimes 2}S_R\right)(U_S^\dagger)^{\otimes 2}(\tilde{\mathcal N}_{S\to E}^{\dagger})^{\otimes 2}(S_E)\right)\\
    &=\tr_S\bigg((U^{(s)}_{A_s}...U^{(1)}_{A_1}U^{(-1)}_{\{1\}}...U^{(-n)}_{\{n\}})^{\otimes 2}\tr_R\left(\rho_{SR}^{\otimes 2}S_R\right)\left(\left(U^{(s)}_{A_s}\dots U^{(1)}_{A_1}U^{(-1)}_{\{1\}}\dots U^{(-n)}_{\{n\}}\right)^\dagger\right)^{\otimes 2}\notag\\
    &\hspace{6em}\times\left(\left(V_{\{1\}}^{(-1)}\dots V_{\{n\}}^{(-n)}\right)^\dagger\right)^{\otimes 2}\left(\tilde{\mathcal N}^\dagger_{S\to E}\right)^{\otimes 2}(S_E)\left(V_{\{1\}}^{(-1)}\dots V_{\{n\}}^{(-n)}\right)^{\otimes 2}\bigg).
\end{align}
Using linearity of the expectation to move it into the trace yields
\begin{equation}
    E_{U_S\sim \chi}\left[\tr\left(\left(\tilde{\mathcal N}_{S\to E}(U_S\rho_{SR}U_S^\dagger)\right)^2\right)\right]=\tr\left(\E_{U_S\sim \chi}\left[\left(\tilde{\mathcal N}_{S\to E}(U_S\rho_{SR}U_S^\dagger)\right)^2\right]\right).
\end{equation}
Omitting the trace and focusing on the expectation gives
{\allowdisplaybreaks
\begin{align}
    \frac{1}{2^k}Z&=\E_{U_S\sim \chi}\left[\left(\tilde{\mathcal N}_{S\to E}(U_S\rho_{SR}U_S^\dagger)\right)^2\right]\\
    &=\E_{\substack{\{U^{(i)}_{A_i}\}\sim \left(\chi_{\mathcal C}^{(2)}\right)^{\times m}\\ \{U^{(-i)}_{\{i\}}\}\sim\left(\Haar(2)\right)^{\times n}}}\left[(U^{(s)}_{A_s}\dots U^{(-n)}_{\{n\}})^{\otimes 2}\tr_R\left(\rho_{SR}^{\otimes 2}S_R\right)\left(\left(U^{(s)}_{A_s}\dots U^{(-n)}_{\{n\}}\right)^\dagger\right)^{\otimes 2}\right] \\
    &\qquad \times \E_{\{V_{\{i\}}^{(-i)}\}\sim \left(\Haar(2)\right)^{\times n}}\left[\left(\left(V_{\{1\}}^{(-1)}\dots V_{\{n\}}^{(-n)}\right)^\dagger\right)^{\otimes 2}\left(\tilde{\mathcal N}^\dagger_{S\to E}\right)^{\otimes 2}(S_E)\left(V_{\{1\}}^{(-1)}\dots V_{\{n\}}^{(-n)}\right)^{\otimes 2}\right]. \nonumber
\end{align}}Now, let $M^{(-i)}(\rho)=\E_{U\sim \Haar(2)}\left[U^{\otimes 2}\rho (U^\dagger)^{\otimes 2}\right]$ and $M^{(i)}(\rho)=\E_{U\sim \Haar(4)}\left[U^{\otimes 2}\rho (U^\dagger)^{\otimes 2}\right]$. Using that uniform sampling from Clifford gates forms a unitary 2-design gives us
\begin{align}\label{eq:starttrajectory}
    \frac{1}{2^k}Z&=\E_{U_S\sim \chi}\left[\left(\tilde{\mathcal N}_{S\to E}(U_S\rho_{SR}U_S^\dagger)\right)^2\right]\\
    &=\left(M^{(s)}\circ \dots M^{(1)}\circ M^{(-1)}\circ \dots \circ M^{(-n)}\left(\tr_R\left(\rho_{SR}^{\otimes 2}S_R\right)\right)\right) \left(M^{(-n)}\circ \dots \circ M^{(-1)}\left(\left(\tilde{\mathcal N}^\dagger_{S\to E}\right)^{\otimes 2}(S_E)\right)\right).\nonumber
\end{align}

Let us focus on the first factor in Equation \eqref{eq:starttrajectory}. As $M^{(-i)}$, for $i>0$, acts on the 2-copies of the $i$-th state for any $\rho$, we find that
\begin{equation}M^{(-i)}(\rho_{[n]\setminus \{i\}}\otimes \rho_{\{i\}})=\rho_{[n]\setminus \{i\}}\otimes (\alpha_i I+\beta_i S)_{\{i\}}.\end{equation}
Using Lemma~\ref{lemma:haartwirl}, we find that when $i$ corresponds to a logical qubit that $\rho_{SR\{i\}}=|\hat \phi\rangle\langle \hat \phi|_{L_iR_i}=\rho_{LR\{i\}}$, where $L_iR_i$ corresponds to the registers of qubit $i$, and thus
\begin{align}
    \alpha_i&=\frac{1}{3}\tr\left(\tr_{R_i}\left(\rho_{LR\{i\}}^{\otimes 2}S_{R_i}\right)\right)-\frac{1}{6}\tr\left(\tr_{R_i}\left(\rho_{LR\{i\}}^{\otimes 2}S_{R_i}\right)S_{L_i}\right)\\
    &=\frac{1}{3}\tr\left(\tr_{L_i}\left(\rho_{LR\{i\}}^{\otimes 2}\right)S_{R_i}\right)-\frac{1}{6}\tr\left(\rho_{LR\{i\}}^{\otimes 2}(S_{L_i}\otimes S_{R_i})\right)\\
    &=\frac13\tr\left(\rho_{R\{i\}}^{\otimes 2}S_{R_i}\right)-\frac16 \tr\left(\rho_{LR\{i\}}^2\right)\\
    &=\frac13\tr\left(\rho_{R\{i\}}^2\right)-\frac16\\
    &=\frac13 \cdot \frac12 -\frac16\\
    &=0,
\end{align}
where we used that $\rho_{R\{i\}}=\frac12 I$ is the maximally mixed state.
From this, using Lemma~\ref{lemma:haartwirl}, $\beta_i$ can be determined as
\begin{equation}
    \beta_i=\frac13-\frac16\cdot \frac12=\frac14.
\end{equation}
If qubit at position $i$ is an ancilla qubit, we get $\rho_{SR\{i\}}=|0\rangle\langle 0|$ giving us
\begin{equation}\alpha_i=\frac13 \tr\left(|0\rangle\langle0|^{\otimes 2}\right)-\frac16 \tr\left(|0\rangle\langle 0|^{\otimes 2}S\right)=\frac13-\frac16=\frac16.
\end{equation}
Similarly, $\beta_i$ can be determined to give: $\beta_i=\frac16$. Putting it all together, brings us to
\begin{align}
    M^{(-1)}\circ \dots \circ M^{(-n)}\left(\tr_R\left(\rho_{SR}^{\otimes 2}S_R\right)\right)&=\bigotimes_{\ell=1}^m\left(\frac14 S\right)^{\otimes a}\otimes \left(\frac16(I+S)\right)^{\otimes (b-a)}\\
    &=\frac{3^k}{2^k\cdot 6^n}\sum_{\gamma\in (\{S\}^a\times \{I,S\}^{b-a})^m}\bigotimes_{i=1}^n \gamma_i.\label{eq:layer1trajectory}
\end{align}

For the other factor in Equation \eqref{eq:starttrajectory}, let $\sigma =\left(\tilde{\mathcal N}^\dagger_{S\to E}\right)^{\otimes 2}(S_E)\in \mathcal L(\mathcal H_S)$. Using Lemma~\ref{lemma:haartwirl} again gives $M^{(-i)}(\sigma)=\sigma_{[n]\setminus \{i\}}\otimes (\alpha_i'I+\beta_i'S)_{\{i\}}$. Splitting the $S$ and $E$ registers up into $n$ regions, one for each qubit and denoting it as $S_i, E_i$ for $i\in [n]$, such that $\hat{\mathcal N}_{S\to E}=\bigotimes_{i=1}^n\hat{\mathcal N}_{S_i\to E_i}$ (since we only consider Pauli and erasure noise channels this is valid as they factor over each qubit) we note that $\tau_{SE}=\bigotimes_{i=1}^n\tau_{S_iE_i}$, where $\tau_{SE}=(I_S\otimes\hat{\mathcal N}_{S'\to E})(|\hat \phi\rangle \langle \hat \phi|_{SS'})$. Now, determining $\alpha_i'$ gives:
\begin{align}
    \alpha_i'&=\frac13\tr\left(\left(\tilde{\mathcal N}^\dagger_{S_i\to E_i}\right)^{\otimes 2}(S_{E_i})\right)-\frac16\tr\left(\left(\tilde{\mathcal N}^\dagger_{S_i\to E_i}\right)^{\otimes 2}(S_{E_i})S_{S_i}\right).
\end{align}
These terms are determined in \cite{liu2025approximatequantumerrorcorrection} in Equations (B18, B19, B33, B66) and (B88). Combining this, one finds
\begin{equation}
    \tr\left(\left(\tilde{\mathcal N}^\dagger_{S_i\to E_i}\right)^{\otimes 2} (S_{E_i})\right)=4, \qquad  \tr\left(\left(\tilde{\mathcal N}^\dagger_{S_i\to E_i}\right)^{\otimes 2}(S_{E_i})S_{S_i}\right)=4\cdot 2^{f-1}
\end{equation}
More specifically, for the Pauli noise, we have that $f=f_P(\vec p)$ while for the erasure noise we have that $f=f_e(p)$. Instead, we shall use $f$ to denote the \say{strength} of the error and observing that $0\leq f_e(p)\leq 2$ and $0\leq f_P(\vec p)\leq 2$, we constrain $0\leq f\leq 2$.  Putting it all together gives
\begin{equation}
    \alpha':=\alpha_i'=\frac{4}{3}-\frac{2}{3}2^{f-1}, \qquad \beta':=\beta_i'=\frac{4}{3}2^{f-1}-\frac{2}{3}.
\end{equation}

Next, we examine the action of $M^{(t)}$ for $i>0$ after the single-qubit gates have been applied to the first factor in Equation \eqref{eq:starttrajectory}. We follow the reasoning in \cite{Dalzell_2022}. In this case, $U^{(t)}$ is a $4\times 4$ matrix and acts on the qubit pair $A_t$. We can still apply Lemma~\ref{lemma:haartwirl} by replacing $I$ with $I\otimes I$ and $S$ with $S\otimes S$ which is the swap operation on the two copies of the two qubits. We have:
\begin{align}
    \space&M^{(t)}\left(\rho_{[n]\setminus A_t}\otimes (I\otimes I)\right)=\rho_{[n]\setminus A_t}\otimes (I\otimes I)_{A_t}\\
    &M^{(t)}\left(\rho_{[n]\setminus A_t}\otimes (S\otimes S)\right)=\rho_{[n]\setminus A_t}\otimes (S\otimes S)_{A_t}\\
    &M^{(t)}\left(\rho_{[n]\setminus A_t}\otimes (I\otimes S)\right)=M^{(t)}\left(\rho_{[n]\setminus A_t}\otimes (S\otimes I)\right)=\frac25\rho_{[n]\setminus A_t}\otimes\left(I\otimes I_{A_t} + S\otimes S_{A_t}\right).
\end{align}
Thus, if $\rho$ is a linear combination of $\{I,S\}^n$, then $M^{(t)}(\rho)$ will also be a linear combination of terms in $\{I,S\}^n$. Defining for $\gamma,\nu \in \{I,S\}$, $M^{(t)}_{\nu\gamma}$ such that
\begin{equation}M^{(t)}\left(\bigotimes_{i=1}^n\gamma_j\right)=\sum_{\nu\in \{I,S\}^n}M^{(t)}_{\nu\gamma}\bigotimes_{i=1}^n \nu_i,\end{equation}
we recover the exact definition of $M^{(t)}_{\nu\gamma}$ from Equation \eqref{eq:momentmatrix}. 

Combining all of the above, we see that after the first layer, the state will be a linear combination of states in $(\{S\}^a\times\{I,S\}^{b-a})^m$. Not only that, but also the action of $M^{(t)}$ ensures that the state will remain a linear combination of terms in $\{I,S\}^n$ for all $t>0$. Let $c\in \{0,1\}^n$ and write $S^c:=\bigotimes_{i=1}^nS^{c_i}$ for a term that appears in that final state. Then, the unnormalized contribution to $Z$ is:
\begin{align}
    \tr\left(S^c\otimes \bigotimes_{i=1}^n(\alpha'I+\beta'S)\right)&=\tr\left( \bigotimes_{i=1}^n(\alpha'S^{c_i}+\beta'S^{c_i+1})\right)\\
    &=\prod_{i=1}^n\left(\frac{4\alpha'}{2^{c_i}}+2\beta'2^{c_i}\right)\\
    &=(2\alpha'+4\beta')^{|c_i|}(4\alpha'+2\beta')^{n-|c_i|}\\
    &=(4\alpha'+2\beta')^n\cdot \left(\frac{2\alpha'+4\beta'}{4\alpha'+2\beta'}\right)^{|c_i|}\\
    &=4^n\cdot \lambda^{|c_i|}.\label{eq:finaltrajectory}
\end{align}
Hence, for the partition function, combining Equations \eqref{eq:starttrajectory} and \eqref{eq:layer1trajectory} gives:
\begin{align}
    Z&=2^k\E_{U\sim \chi}\left[\tr_S\left(\left(\tilde{\mathcal N}(U_S\rho_{SR}U_S^\dagger)\right)^2\right)\right]\\
    &=2^k\cdot \frac{3^k}{2^k\cdot 6^n}\tr\left(\left(M^{(s)}\circ M^{(1)}\left(\sum_{\gamma\in (\{S\}^a\times \{I,S\}^{b-a})^m}\bigotimes_{i=1}^n \gamma_i\right)\right)\otimes\left(\alpha'I+\beta'S\right)^{\otimes n}\right)\\
    &=\frac{3^k}{6^n}\sum_{\substack{\vec \gamma\in \{I,S\}^{n\times (s+1)}\\ \vec \gamma^{(0)}\in (\{S\}^a\times\{I,S\}^{b-a})^m}}4^n\lambda^{|\vec\gamma^{(s)}|}\prod_{t=1}^sM^{(t)}_{\vec \gamma^{(t)}\vec \gamma^{(t-1)}}\label{eq:laststep}\\
    &=3^k\frac{2^n}{3^n}\sum_{\substack{\vec \gamma\in \{I,S\}^{n\times (s+1)}\\ \vec \gamma^{(0)}\in (\{S\}^a\times\{I,S\}^{b-a})^m}}\lambda^{|\vec\gamma^{(s)}|}\prod_{t=1}^sM^{(t)}_{\vec \gamma^{(t)}\vec \gamma^{(t-1)}},
\end{align}
where in Equation \eqref{eq:laststep} we used that the contribution of each term in the final state is determined by Equation \eqref{eq:finaltrajectory}. 
\end{proof}

\section{Missing proofs of Section~\ref{sec:EC}}
\begin{proof}[Proof of Lemma~\ref{lem:partition-EC}]
We had already found the following bound on the probability of failure:
\begin{align}\mathbb P(F)\leq \sum_{\substack{\nu_A\in \{I,X,Y ,Z\}^{k}\\\nu_A\neq I^{\times k}}}\sum_{\nu_B\in \{I,Z\}^{n-k}}\sum_{\substack{\mu\in \{I,X,Y,Z\}^n\\ 1\leq \omega(\mu)\leq d}}\mathbb P\left(\tr\left(\sigma_\mu U((\sigma_{\nu_A})_L\otimes (\sigma_{\nu_B})_{S\setminus L})U^\dagger\right)\neq 0\right).
    \end{align}

Notice that $\left|\tr\left(\sigma_\mu U((\sigma_{\nu_A})_L\otimes (\sigma_{\nu_B})_{S\setminus L})U^\dagger\right)\right|= 2^n$ whenever it is nonzero as $U\in \mathcal C_n$. Thus,  the expression $\frac{1}{2^n}\left|\tr\left(\sigma_\mu U((\sigma_{\nu_A})_L\otimes (\sigma_{\nu_B})_{S\setminus L})U^\dagger\right)\right|$ is an indicator function and using this yields

    \begin{align}\mathbb P\left(\tr\left(\sigma_\mu U((\sigma_{\nu_A})_L\otimes (\sigma_{\nu_B})_{S\setminus L})U^\dagger\right)\neq 0\right)&=\mathbb P\left(\frac{1}{4^n}\left(\tr\left(\sigma_\mu U((\sigma_{\nu_A})_L\otimes (\sigma_{\nu_B})_{S\setminus L})U^\dagger\right)\right)^2\neq 0\right)\\
    &=\frac{1}{4^n}\E_{U\sim \chi_{\brick}^{(D)}}\left[\left(\tr\left(\sigma_\mu U((\sigma_{\nu_A})_L\otimes (\sigma_{\nu_B})_{S\setminus L})U^\dagger\right)\right)^2\right]\\
    &=\frac{1}{4^n}\E_{U\sim \chi_{\brick}^{(D)}}\left[\tr\left(\sigma_\mu^{\otimes 2}U^{\otimes 2}((\sigma_{\nu_A})_L\otimes (\sigma_{\nu_B})_{S\setminus L})^{\otimes 2}(U^\dagger)^{\otimes 2}\right)\right]\\
    &=\frac{1}{4^n}\tr\left(\sigma_\mu^{\otimes 2}\E_{U\sim \chi_{\brick}^{(D)}}\left[U^{\otimes 2}((\sigma_{\nu_A})_L\otimes (\sigma_{\nu_B})_{S\setminus L})^{\otimes 2}(U^\dagger)^{\otimes 2}\right]\right).
    \end{align}
    Plugging this into the sum and using linearity of trace and expectation, we find
    \begin{align}\mathbb P(F)&\leq \sum_{\substack{\nu_A\in \{I,X,Y ,Z\}^{k}\\\nu_A\neq I^{\times k}}}\sum_{\nu_B\in \{I,Z\}^{n-k}}\sum_{\substack{\mu\in \{I,X,Y,Z\}^n\\ 1\leq \omega(\mu)\leq d}}\mathbb P(\tr\left(\sigma_\mu U((\sigma_{\nu_A})_L\otimes (\sigma_{\nu_B})_{S\setminus L})U^\dagger\right)\neq I^{\times k})\\
    &=\frac{1}{4^n}\sum_{\substack{\nu_A\in \{I,X,Y,Z\}^{k}\\\nu_A\neq I^{\times k}}}\sum_{\nu_B\in \{I,Z\}^{n-k}}\sum_{\substack{\mu\in \{I,X,Y,Z\}^n\\ 1\leq \omega(\mu)\leq d}} \tr\left(\sigma_\mu^{\otimes 2}\E_{U\sim \chi_{\brick}^{(D)}}\left[U^{\otimes 2}((\sigma_{\nu_A})_L\otimes (\sigma_{\nu_B})_{S\setminus L})^{\otimes 2}(U^\dagger)^{\otimes 2}\right]\right).
    \end{align}

    Let us first focus on the term 
\begin{equation}\E_{U\sim \chi_{\brick}^{(D)}}\left[U^{\otimes 2}((\sigma_{\nu_A})_L\otimes (\sigma_{\nu_B})_{S\setminus L})^{\otimes 2}(U^\dagger)^{\otimes 2}\right].\end{equation}
We define 
\begin{equation}M^{(t)}(\rho)=\begin{cases}
    \E_{U\sim \Haar(2)}\left[U^{\otimes 2}\rho(U^\dagger)^{\otimes 2}\right],& \text{if $t<0$}, \\
    \E_{U\sim \Haar(4)}\left[U^{\otimes 2}\rho(U^\dagger)^{\otimes 2}\right],& \text{if $t>0$}.
\end{cases}
\end{equation}
where $\Haar(2)$ and $\Haar(4)$ are the Haar random distributions on $1$ and $2$ qubits, respectively.
Note that as each $U^{(j)}_{A_j}$ in $U=U^{(s)}_{A_s}...U^{(1)}_{A_1}$ is sampled uniformly over the set of 2-qubit Clifford gates, we can assume without loss of generality that there is a layer of Haar random single-qubit gates, similar to the argument in the proof of Lemma~\ref{lem:partition}. Using that the Clifford gates themselves are exact unitary $2$-designs then gives
\begin{equation}
\E_{U\sim \chi_{\brick}^{(D)}}\left[U^{\otimes 2}((\sigma_{\nu_A})_L\otimes (\sigma_{\nu_B})_{S\setminus L})^{\otimes 2}(U^\dagger)^{\otimes 2}\right]=M^{(s)}\circ ... \circ M^{(1)}\circ M^{(-1)}\circ ... M^{(-n)}\left(((\sigma_{\nu_A})_L\otimes (\sigma_{\nu_B})_{S\setminus L})^{\otimes 2}\right).
\end{equation}
As $M^{(-j)}$, for $j>0$ acts on the 2-copies of the $j$-th state for any $\rho$, we find that
\begin{equation}
M^{(-j)}(\rho_{[n]\setminus \{j\}}\otimes \rho_{\{j\}})=\rho_{[n]\setminus \{j\}}\otimes (\alpha I+\beta S)_{\{j\}},
\end{equation}
for some $\alpha,\beta \in \mathbb C$ depending on $\rho$ as described in Lemma \ref{lemma:haartwirl}. Thus, letting $(\sigma_{\nu_A})_L\otimes (\sigma_{\nu_B})_{S\setminus L} = \sigma_1\otimes... \otimes \sigma_n$ we compute
\begin{equation}M^{(-j)}(((\sigma_{\nu_A})_L\otimes (\sigma_{\nu_B})_{S\setminus L})^{\otimes 2})=((\sigma_{\nu_A})_L\otimes(\sigma_{\nu_B})_{S\setminus L})^{\otimes 2}_{[n]\setminus \{j\}} \otimes(\alpha I +\beta S),\end{equation}
where we need to determine $\alpha$ and $\beta$. Using $\tr(S\sigma_j^{\otimes 2})=2$ and $\tr(\sigma_j^{\otimes 2})=0$ for non-identity Pauli gates, we find that
\begin{equation}\E_{U\sim \Haar(2)}\left[U^{\otimes 2}\sigma_j^{\otimes 2}(U^\dagger)^{\otimes 2}\right]=\begin{cases} I_4,& \text{if $\sigma_j=I_2$}\\-\frac{1}{3}I_4+\frac23S, & \text{if $\sigma_j\neq I_2$,}\end{cases} \end{equation}
where $I_n$ is the $n\times n$ identity matrix.
With this in mind, we determine
\begin{equation}\sum_{\substack{\nu_A\in \{I,X,Y,Z\}^{k}\\\nu_A\neq 0}}\sum_{\nu_B\in \{I,Z\}^{n-k}} M^{(-1)}\circ ... M^{(-n)}(((\sigma_{\nu_A})_L\otimes (\sigma_{\nu_B})_{S\setminus L})^{\otimes 2}),\end{equation}
as follows. Note that $M^{(-j)}$ only depends on $\sigma_j$ and results in $I_4$, when $\sigma_j=I_2$ or $Q:=-\frac13I+\frac23S$ whenever $\sigma_j\neq I_2$. For a given string $z\in \{0,1\}^k$ with $\omega(z)>0$, and writing $Q^z=\bigotimes_{i=1}^k Q^{z_i}$, we find that there are $3^{w(z)}$ number of $\nu_A\in \{I,X,Y,Z\}^k$ with $\nu_A\neq I^{\times k}$ such that

\begin{equation}\bigotimes_{j\in L}M^{(-j)}(((\sigma_{\nu_A})_L\otimes (\sigma_{\nu_B})_{S\setminus L})^{\otimes 2})=(Q^z)_L\otimes(\sigma_{\nu_B})^{\otimes 2}_{S\setminus L}.\end{equation} 
Furthermore, any sequence in $\{I,Q\}^k$, except for $I^{\otimes k}$ will be ``reached" because $\nu_A\neq I^{\times k}$. Using a similar argument for the $\nu_B$, we find that
\begin{align}\sum_{\substack{\nu_A\in \{I,X,Y,Z\}^{k}\\\nu_A\neq  I^{\times k}}}&\sum_{\nu_B\in \{I,Z\}^{n-k}} M^{(-1)}\circ ... M^{(-n)}(((\sigma_{\nu_A})_L\otimes (\sigma_{\nu_B})_{S\setminus L})^{\otimes 2})\\
&=\left(\sum_{z\in \{0,1\}^k}3^{\omega(a)}(Q^z)_L-I_L\right)\otimes \sum_{z'\in \{0,1\}^{n-k}}(Q^{z'})_{S\setminus L}.
\end{align}
Let us first focus on the first factor. Letting $W=-I+2S$, we have that $3^{\omega(z)}Q^z = W^z$ for any $z\in \{0,1\}^k$. Furthermore, we determine
\begin{equation}\sum_{z\in \{0,1\}^k}W^z = \sum_{z\in \{0,1\}^{k-1}}W^z\otimes I+\sum_{z\in \{0,1\}^{k-1}}W^z\otimes W=2\sum_{z\in \{0,1\}^{k-1}}W^z\otimes S=2^k S^{\otimes k}\end{equation}
Using a similar strategy for the factor acting on ancilla qubits in $S\setminus L$, we find
\begin{equation}\sum_{z'\in \{0,1\}^{n-k}} Q^{z'}=\left(\frac{2}{3}\right)^{n-k}(I+S)^{\otimes (n-k)}=\left(\frac{2}{3}\right)^{n-k}\sum_{\gamma\in \{I,S\}^{n-k}}\bigotimes_{j=1}^{n-k}\gamma_j.\end{equation}
Putting this all together, and writing $O_{S\setminus L}:=\left(\sum_{\gamma\in \{I,S\}^{n-k}}\bigotimes_{j=1}^{n-k}\gamma_j\right)_{S\setminus L}$, we find that
\begin{equation}\label{eq:2-qubitgatemoments}\mathbb P(F)\leq \frac{1}{4^n}\left(\frac{2}{3}\right)^{n-k}\sum_{\substack{\mu \in \{I,X,Y,Z\}^n\\ 1\leq \omega(\mu)\leq d}}\tr\left(\sigma_{\mu}^{\otimes 2}\left(M^{(m)}\circ ... \circ M^{(1)}\left((2^kS_L-I_L)\otimes O_{S\setminus L}\right)\right)\right).\end{equation}
Now, for a given $\gamma\in \{I,S\}^n$ and representing it by $S^c$ for some $c\in \{0,1\}^n$, we determine
\begin{equation}\sum_{\substack{\mu \in \{I,X,Y,Z\}^n\\ 1\leq \omega(\mu)\leq d}}\tr(\sigma_{\mu}^{\otimes 2}S^c)=\sum_{j=1}^d\sum_{\omega(\mu)=j}\tr(\sigma_{\mu}^{\otimes 2}S^c)=\sum_{j=1}^d\sum_{\omega(\mu)=j}\prod_{i=1}^n\tr(\mu_{\mu_i}^{\otimes 2}S^{c_i}).\end{equation}
Since $S^{c_i}$ is either $I$ or $S$, we find that whenever $\sigma_{\mu_i}\neq I$, at these $i$, we must have $c_i=1$ as otherwise the term would evaluate to 0. Given that $\omega(\mu)=j$, we find that $\gamma$ must contain at least $j$ factors equal to $S$ in the tensor product in order to contribute. So, suppose that $\omega(c)=m$. Then if $\omega(\mu)=j$, there are $3^j\binom{m}{j}$ number of $\mu$ such that $\tr(\sigma_{\mu}^{\otimes 2}S^c)\neq 0$. Notice that for each such $\mu$, we have that
\begin{equation}\label{eq:traceweightdependent}\tr(\sigma_{\mu}^{\otimes 2}S^c)=2^{\omega(c)}4^{n-\omega(c)}.\end{equation}
Using this, we evaluate
\begin{align}\label{eq:sumweightedtrace}\sum_{\substack{\mu \in \{I,X,Y,Z\}^n\\ 1\leq \omega(\mu)\leq d}}\tr(\sigma_{\mu}^{\otimes 2}S^c)&=4^n\cdot 2^{-\omega(c)}\sum_{j=1}^{\min\{\omega(c),d\}}3^j\binom{\omega(c)}{j}
\end{align}
Combining Equations~\eqref{eq:2-qubitgatemoments} and \eqref{eq:sumweightedtrace} gives the following upper bound on $\mathbb P(F)$:
\begin{align}\label{eq:lemmainfinitedepthrequiredEC}\mathbb P(F)&\leq \frac{1}{4^n}\left(\frac{2}{3}\right)^{n-k}\sum_{\substack{\mu \in \{I,X,Y,Z\}^n\\ 1\leq \omega(\mu)\leq d}}\tr\left(\sigma_{\mu}^{\otimes 2}\left(M^{(m)}\circ ... \circ M^{(1)}\left((2^kS_L-I_L)\otimes O_{S\setminus L}\right)\right)\right)\\
&=2^k\left(\frac{2}{3}\right)^{n-k}\sum_{\substack{\vec \gamma\in \{I,S\}^{n\times (s+1)}\\ \vec \gamma^{(0)}\in (\{S\}^a\times\{I,S\}^{b-a})^m}}\left(2^{-|\vec \gamma^{(s)}|}\sum_{j=1}^{\min\{d,|\vec \gamma^{(s)}|\}}3^j\binom{|\vec \gamma^{(s)}|}{j}\right)\prod_{t=1}^sM^{(t)}_{\vec \gamma^{(t)}\vec \gamma^{(t-1)}}\\
&\qquad - \left(\frac{2}{3}\right)^{n-k}\sum_{\substack{\vec \gamma\in \{I,S\}^{n\times (s+1)}\\ \vec \gamma^{(0)}\in (\{S\}^a\times\{I,S\}^{b-a})^m}}\left(2^{-|\vec \gamma^{(s)}|}\sum_{j=1}^{\min\{d,|\vec \gamma^{(s)}|\}}3^j\binom{|\vec \gamma^{(s)}|}{j}\right)\prod_{t=1}^sM^{(t)}_{\vec \gamma^{(t)}\vec \gamma^{(t-1)}},\nonumber\\
&\leq 3^k\left(\frac{2}{3}\right)^{n}\sum_{\substack{\vec \gamma\in \{I,S\}^{n\times (s+1)}\\ \vec \gamma^{(0)}\in (\{I\}^a\times\{I,S\}^{b-a})^m}}\left(2^{-|\vec \gamma^{(s)}|}\sum_{j=1}^{\min\{d,|\vec \gamma^{(s)}|\}}3^j\binom{|\vec \gamma^{(s)}|}{j}\right)\prod_{t=1}^sM^{(t)}_{\vec \gamma^{(t)}\vec \gamma^{(t-1)}},
\end{align}
where $M^{(t)}_{\vec \gamma^{(t)}\vec \gamma^{(t-1)}}$ is defined exactly the same as in Lemma~\ref{lem:partition}.
\end{proof}

\begin{proof}[Proof of Lemma~\ref{lem:infinitedepthexact}]
   At infinite depth, we have that $\chi_{\brick}^{(D)}$ becomes the Haar random distribution $\Haar(2^n)$. Hence, we may assume w.l.o.g. that in the first layer there are Haar random single-qubit gates since these can be absorbed into the Haar random gate that is the circuit at infinite depth. In that case, we may continue from Equation~\eqref{eq:lemmainfinitedepthrequiredEC} and see that
   \begin{equation}
       Z^{(\infty)}_{\QEC}=\frac{2^k}{4^n}\left(\frac{2}{3}\right)^{n-k}\sum_{\substack{\mu \in \{I,X,Y,Z\}^n\\ 1\leq \omega(\mu)\leq d}}\tr\left(\sigma_{\mu}^{\otimes 2}\left(\E_{U\sim \Haar(2^n)}\left[U^{\otimes 2}\left(S_L\otimes O_{S\setminus L}\right)(U^\dagger)^{\otimes 2}\right]\right)\right).
   \end{equation}
   Using Lemma~\ref{lemma:haartwirl} we again find that $\E_{U\sim \Haar(2^n)}\left[U^{\otimes 2}\left(S_L\otimes O_{S\setminus L}\right)(U^\dagger)^{\otimes 2}\right]=\alpha I +\beta S$. However, because of the $\sigma_{\mu}^{\otimes 2}$ factor, we may ignore that $\alpha I$ term because $\tr(\sigma_{\mu}^{\otimes 2}I)=0$ as $\omega(\mu)\geq 1$. Thus, continuing with only $\beta S$ then gives
   \begin{align}
       Z_{\QEC}^{(\infty)}&=\beta\cdot\frac{2^k}{4^n}\left(\frac{2}{3}\right)^{n-k}\sum_{\substack{\mu \in \{I,X,Y,Z\}^n\\ 1\leq \omega(\mu)\leq d}}\tr\left(\sigma_{\mu}^{\otimes 2}S\right)\\
       &=\beta\cdot\frac{2^k}{4^n}\left(\frac{2}{3}\right)^{n-k}\sum_{\substack{\mu \in \{I,X,Y,Z\}^n\\ 1\leq \omega(\mu)\leq d}}2^n\\
       &=\beta\cdot 2^n\cdot\frac{2^k}{4^n}\left(\frac{2}{3}\right)^{n-k}\sum_{j=1}^{d}3^j\binom{n}{j},
   \end{align}
   where in the second equality we used Equation~\eqref{eq:traceweightdependent} specialized to $\omega(c)=1$ and the last equality follows by noting that there are $3^j\binom{n}{j}$ number of strings $\mu\in \{I,X,Y,Z\}^n$ of weight $j$. By Lemma~\ref{lemma:haartwirl} we find for $\beta$ the following:
   \begin{align}
       \beta &= \frac{1}{2^{2n}-1}\tr(S_L\otimes O_{S\setminus L} S)-\frac{1}{2^n(2^{2n}-1)}\tr(S_L\otimes O_{S\setminus L})\\
       &=\frac{1}{2^{2n}-1}\tr(I_L\otimes O_{S\setminus L})-\frac{1}{2^n(2^{2n}-1)}\tr(S_L\otimes O_{S\setminus L})\\
       &=\frac{4^k\cdot 6^{n-k}}{2^{2n}-1} - \frac{2^k\cdot 6^{n-k}}{2^n(2^{2n}-1)}\\
       &=\frac{4^k\cdot 6^{n-k}\cdot 2^n-2^k\cdot 6^{n-k}}{2^n(2^{2n}-1)}.
   \end{align}
   Combining everything then gives
   \begin{equation}
       Z^{(\infty)}_{\QEC}=\frac{4^k\cdot 6^{n-k}\cdot 2^n-2^k\cdot 6^{n-k}}{2^n(2^{2n}-1)}\cdot 2^n\cdot \frac{2^k}{4^n} \left(\frac{2}{3}\right)^{n-k}\cdot K=\frac{K}{1-2^{-2n}}(2^{-n+k}-2^{-2n}),
   \end{equation}
   which proves the statement.
\end{proof}
\begin{proof}[Proof of Lemma~\ref{lemma:globalboundEC}]
    For $|\vec \gamma^{(s)}|\leq d$ we immediately find the following bound
\begin{equation}\label{eq:upperboundlessthand}
    2^{-f|\vec \gamma^{(s)}|}\sum_{j=1}^{\min\{d,|\vec \gamma^{(s)}|\}}3^j\binom{|\vec \gamma^{(s)}|}{j}=2^{-f|\vec \gamma^{(s)}|}\sum_{j=1}^{|\vec \gamma^{(s)}|}3^j\binom{|\vec \gamma^{(s)}|}{j}=2^{-f|\vec \gamma^{(s)}|}\cdot 4^{|\vec \gamma^{(s)}|}\leq 2^{(2-f)d}.
\end{equation}
When $|\vec \gamma^{(s)}|\geq d$ we use the following bounds
\begin{equation}\label{eq:upperboundgreaterthand}
    2^{-f|\vec \gamma^{(s)}|}\sum_{j=1}^{\min\{d,|\vec \gamma^{(s)}|\}}3^j\binom{|\vec \gamma^{(s)}|}{j}=2^{-f|\vec \gamma^{(s)}|}\sum_{j=1}^{d}3^j\left(\frac{e |\vec \gamma^{(s)}|}{j}\right)^j\leq d\cdot 2^{-f|\vec\gamma^{(s)}|}\cdot 3^d\cdot \left(\frac{e |\vec \gamma^{(s)}|}{d}\right)^d,
\end{equation}
where in the last inequality we used that $3^j\left(\frac{e |\vec \gamma^{(s)}|}{j}\right)^j$ is an increasing function in $j$. Next, we consider the function $h(z):= 2^{-fz}\left(\frac{ez}{d}\right)^d$ and find that with some differential calculus it attains its maximum at $z^\ast = \frac{d}{f\ln(2)}$ and see that 
\begin{equation}h(z^\ast)=2^{-\frac{d}{\ln(2)}}\left(\frac{e}{f\ln(2)}\right)^d=\left(\frac{1}{f\ln(2)}\right)^d.\end{equation}
Combining this with the inequality in Equation~\eqref{eq:upperboundgreaterthand} gives
\begin{equation}
    2^{-f|\vec \gamma^{(s)}|}\sum_{j=1}^{\min\{d,|\vec \gamma^{(s)}|\}}3^j\binom{|\vec \gamma^{(s)}|}{j}\leq d\left(\frac{3}{f\ln(2)}\right)^d.
\end{equation}
Further combining this with the inequality in Equation~\eqref{eq:upperboundlessthand} then immediately gives
\begin{equation}
    2^{-f|\vec \gamma^{(s)}|}\sum_{j=1}^{\min\{d,|\vec \gamma^{(s)}|\}}\leq d\cdot A^d,
\end{equation}
for all $1\leq |\vec\gamma^{(s)}|\leq n$, as desired.
\end{proof}

\end{document}